\documentclass[english,12pt]{article}

\usepackage{endnotes}
\usepackage{authblk}
\usepackage[T1]{fontenc}
\usepackage[utf8]{inputenc}
\usepackage{geometry}
\usepackage{babel}
\usepackage{amsthm}
\usepackage{amsmath}
\usepackage{amssymb}
\usepackage{mathrsfs}
\usepackage{mathtools}
\usepackage{bbm}
\usepackage{babel}
\usepackage{harvard}
\usepackage{booktabs}
\usepackage{fnpct}
\usepackage{appendix}
\usepackage{pdfpages}
\usepackage[open,openlevel=1]{bookmark}
\usepackage{xcolor}
\hypersetup{
colorlinks,
linkcolor={red!50!black},
citecolor={blue!50!black},
urlcolor={blue!80!black}
}
\usepackage{multirow}
\usepackage{tikz}
\usepackage{rotating}
\usepackage[normalem]{ulem}

\usepackage{thmtools, thm-restate}  % tools for organizing theorem-like environment
\declaretheoremstyle[
spaceabove=6pt, spacebelow=6pt,
headfont=\normalfont\bfseries,
notefont=\mdseries, notebraces={(}{)},
bodyfont=\normalfont,
postheadspace=1em,
qed=$\blacksquare$
]{examplestyle}

\declaretheoremstyle[
spaceabove=6pt, spacebelow=6pt,
headfont=\normalfont\bfseries,
notefont=\mdseries, notebraces={(}{)},
bodyfont=\itshape,
postheadspace=1em
]{theorem}

\declaretheoremstyle[
spaceabove=6pt, spacebelow=6pt,
headfont=\normalfont\bfseries,
notefont=\mdseries, notebraces={(}{)},
bodyfont=\normalfont,
postheadspace=1em
]{assumption}

\declaretheoremstyle[
spaceabove=4pt, spacebelow=4pt,
headfont=\itshape\bfseries,
notefont=\mdseries, notebraces={(}{)},
bodyfont=\itshape,
postheadspace=0.2em,
qed=\qedsymbol
]{remark}

\declaretheorem[style=theorem]{theorem}
\declaretheorem[style=theorem, numbered=no, name=Theorem]{theorem*}

\declaretheorem[style=theorem, name=Proposition]{proposition}
\declaretheorem[numbered=no,style=definition,name=Question]{question*}  % question
\declaretheorem[style=definition,name=Definition, numbered=no]{definition*}  % definition

\declaretheorem[style=examplestyle]{example}
\renewcommand\thmcontinues[1]{continued}

\newtheorem{thm}{Assumption}

\newcommand{\norm}[1]{\left\Vert #1\right\Vert }

\newcommand{\pd}[2]{\frac{\partial#1}{\partial#2}}

\newcommand{\dd}[2]{\frac{\mathrm{d}#1}{\mathrm{d}#2}}

\newcommand{\expt}{\text{\normalfont E}}
\newcommand{\var}{\text{\normalfont Var}}

\newcommand{\corr}{\text{\normalfont corr}}

\newcommand{\convd}{\overset{d}{\to}}

\newcounter{steps}

\DeclareMathOperator*{\argmin}{arg\,min}
\newcommand\independent{\protect\mathpalette{\protect\independenT}{\perp}}
\def\independenT#1#2{\mathrel{\rlap{$#1#2$}\mkern2mu{#1#2}}}

\usepackage[nodisplayskipstretch]{setspace}
\usepackage{enumitem}
\setenumerate[1]{label=(\roman*)}
\setenumerate[2]{label=\normalfont{(\alph*)}}

\title{Production Function Estimation without Invertibility: Imperfectly Competitive Environments and Demand Shocks\thanks{We thank Dan Ackerberg and Steve Berry for helpful discussions and comments and Alexander Whitefield, Vivian Zhang, and Wanxi Zhou for excellent research assistance.}}
\author{Ulrich Doraszelski\\
University of Pennsylvania\thanks{%
Wharton School, Email: doraszelski@wharton.upenn.edu} \and Lixiong Li\\
Johns Hopkins University\thanks{%
Department of Economics, Email: lixiong.li@jhu.edu}}
\date{\today}

\begin{document}
\maketitle

\begin{abstract}
We advance the proxy variable approach to production function estimation. We show that the invertibility assumption at its heart is testable. We characterize what goes wrong if invertibility fails and what can still be done. We show that rethinking how the estimation procedure is implemented either eliminates or mitigates the bias that arises if invertibility fails. In particular, a simple change to the first step of the estimation procedure provides a first-order bias correction for the GMM estimator in the second step. Furthermore, a modification of the moment condition in the second step ensures Neyman orthogonality and enhances efficiency and robustness by rendering the asymptotic distribution of the GMM estimator invariant to estimation noise from the first step.
 \end{abstract}

\section{Introduction\label{sec:introduction}}

Production functions are one of the oldest concepts in economics \cite{VONT:42,WICK:94,COBB:28}. As a description of the relationship between inputs and output, they are of interest in and of themselves and also a vehicle for measuring productivity and technical change \cite{SOLO:57}. More recently, production functions have become a key input into estimating markups and markdowns in order to assess firms' market power \cite{HALL:88,DELO:12}. The production approach to markup estimation recognizes that, given an estimate of the production function, the markup can be recovered from the firm's cost minimization problem. To estimate the production function, this large and rapidly growing literature relies on the procedure initially developed by \citeasnoun{OLLE:96} and further advanced by \citeasnoun{LEVI:03} and \citeasnoun{ACKE:15} (henceforth OP, LP, and ACF).

The OP/LP/ACF procedure starts from the observation that there is an endogeneity problem in estimating production functions. As first pointed out by \citeasnoun{MARS:44}, this problem arises because the decisions that the firm makes regarding its inputs depend on its productivity, which is unobserved by the econometrician. To resolve this problem, OP turn it on its head and use the firm's decisions that are observed by the econometrician to infer its productivity. The resulting proxy variable approach to production function estimation has quickly become standard practice.

The inversion from observables to productivity at the heart of the OP/LP/ACF procedure requires that firms with the same productivity make the same choices. Yet, there is no reason to believe this is the case, e.g., if these firms face different demands in the output market. The current paper therefore revisits the OP/LP/ACF procedure if invertibility fails. It characterizes what goes wrong in the OP/LP/ACF procedure and what can still be done.

The point that the OP/LP/ACF procedure cannot accommodate unobserved demand heterogeneity has been made before. \citeasnoun{FOST:08} put it as follows:
\begin{quote}
\ldots\ idiosyncratic demand shocks make the proxies functions of both
technology and demand shocks, thereby inducing a possible omitted variable
bias. Put simply, proxy methods require a one-to-one mapping between
plant-level productivity and the observables used to proxy for productivity.
This mapping breaks down if other unobservable plant-level factors besides
productivity drive changes in the observable proxy. (p. 403)
\end{quote}
Indeed, to establish invertibility, OP rule out that firms face different demands and abstract from competition between firms.\footnote{OP assume that any profitability differences across firms are due to differences in their capital stocks and productivities (p. 1273), thereby ruling out that firms face different demands. Limiting the state variables in the firm's investment policy to its own capital stock and productivity moreover abstracts from competition between firms \citeaffixed{PAKE:94B}{see also Lemma 3 and Theorem 1 in}.} LP assume a perfectly competitive industry where firms act as price takers and thus face the same horizontal demand curve (p. 322 and Appendix A). Based on their results, ACF start from scalar unobservable and strict monotonicity assumptions, and almost all subsequent literature simply imposes invertibility as a high-level assumption. For this reason, it may not have fully appreciated the demanding nature of the invertibility assumption.

Unobserved demand heterogeneity is likely to be empirically important. The large literatures on demand estimation and productivity analysis highlight the considerable heterogeneity in demand that remains even after controlling for detailed product attributes \cite{BERR:95} or honing in on (nearly) homogenous products \cite{FOST:08}. The problem is compounded by the fact that, in imperfectly competitive environments, a firm's decisions in equilibrium depend on its own productivity as well as on the productivities of its rivals. Hence, imperfectly competitive environments require jointly inverting the decisions of all firms for the productivities of all firms. This many-to-many inverse may not exist \cite{BION:22} or it may be too high-dimensional to be practical \cite{ACKE:24}. Moreover, a firm's rivals are partially or completely unobserved in typical datasets used for production function estimation. The decisions of its rivals have thus to be thought of as shocks to the demand the focal firm faces.

Invertibility can fail for reasons other than unobserved demand heterogeneity. Changes in firm conduct due to mergers and acquisitions or switches from competition to collusion are best thought of as changes to the ownership matrix. Invertibility can fail if these changes in firm conduct are partially or completely unobserved by the econometrician. Finally, invertibility can fail if there is unobserved variation across firms or time in input prices, investment opportunities, or financial constraints.

%Taken together, the invertibility assumption is especially challenging for the production approach to markup estimation for two reasons. First, the demand a firm faces in the output market is a fundamental determinant of the markup that its charges. Second, estimating markups is arguably most important in imperfectly competitive environments and after changes in firm conduct.

In light of the demanding nature of the invertibility assumption, this paper makes five contributions. First, we propose tests for invertibility. Our tests exploit that if invertibility fails, then productivity becomes a hidden state in a Markov model. The literature on Kalman and particle filters in dynamic systems shows that the best guess for the hidden state uses the entire history of the observables as opposed to just their current value. Invertibility can therefore be tested by including lags of the observables in the regression in the first step of the OP/LP/ACF procedure. Implementing our tests on an unbalanced panel of Spanish manufacturing firms and a balanced panel of US manufacturing industries, we strongly reject invertibility.

We therefore characterize the consequences of a failure of invertibility for the OP/LP/ACF procedure. The first step of the OP/LP/ACF procedure regresses output on observables. If invertibility fails, then the prediction of output contains an error. Because the prediction is used to control for lagged productivity in the GMM estimation in the second step, the lagged prediction error enters into the conditional moment. Building on \citeasnoun{DORA:21}, we show that this invalidates capital as an instrument and results in biased estimates. %However, because the prediction error is by construction mean independent of the observables used in the first step, other instruments such as lagged capital, labor, and materials remain valid.

Turning from what goes wrong if invertibility fails to what can still be done, our second contribution is to provide a necessary and sufficient condition for the moment condition in the second step of the OP/LP/ACF procedure to hold for the true production function and some (not necessarily the true) law of motion for productivity. This condition calls for rethinking how the OP/LP/ACF procedure is implemented. In particular, it compels us to ensure that any instrument  used in the second step is appropriately included in the regression in the first step. Due to the lag structure of the model, this means including the lead of capital in addition to its current value in the regression. We provide a series of examples where this simple expedient suffices to satisfy our necessary and sufficient condition.

Beyond these examples, our necessary and sufficient condition may be violated. While this results in biased estimates, our third contribution is to show that rethinking how the OP/LP/ACF procedure is implemented mitigates the bias. Specifically, we show that the moment condition in the second step implicitly incorporates a first-order bias correction provided any instrument used in the second step is appropriately included in the regression in the first step.

Our fourth contribution is to explicitly incorporate a bias correction into the second step of the OP/LP/ACF procedure. We show that the modified moment condition has a property known as Neyman orthogonality \cite{NEYM:59}. While our modification remains as straightforward to implement as the original OP/LP/ACF procedure, Neyman orthogonality renders the asymptotic distribution of the GMM estimator in the second step invariant to estimation noise and the quality of the prediction of output from the first step. This is particularly advantageous if the regression in the first step includes a large number of covariates. Neyman orthogonality facilitates the use of a wide range of estimation methods in the first step, including traditional nonparametric methods as well as modern machine learning techniques such as neural networks and random forests. A Monte Carlo exercise shows that Neyman orthogonality can substantially improve the performance of the GMM estimator in finite samples.

Our fifth and final contribution is to provide a diagnostic to assess the sensitivity of the estimates to the size of the prediction error in the first step. While our diagnostic has some similarities to the sensitivity measure in \citeasnoun{ANDR:17}, it is not local to the true model and directly informative about the estimated model. Our diagnostic is neither necessary nor sufficient for no bias. However, a small value of the diagnostic provides assurance that small changes in the prediction error do not dramatically alter the estimates.

In sum, this paper examines the OP/LP/ACF procedure if the invertibility assumption fails. Invertibility fails if there are demand shocks or in imperfectly competitive environments with partially or completely unobserved rivals or changes in firm conduct. Whether invertibility fails can be tested. A failure of invertibility can have a substantial, detrimental impact on the estimates.

Fortunately, much can still be done. We provide a necessary and sufficient condition for the moment condition in the second step of the OP/LP/ACF procedure to hold for the true production function. This condition compels us to ensure that any instrument used in the second step is appropriately included in the regression in the first step. We show that this simple change either eliminates or mitigates the bias. Going a step further, we modify the moment condition in the second step to endow it with Neyman orthogonality. Finally, we provide a diagnostic to assess the sensitivity of the estimates to the size of the prediction error that arises in the first step of the OP/LP/ACF procedure if invertibility fails.

The remainder of the paper is organized as follows. In Section \ref{sec:model}, we recall the setup and the OP/LP/ACF procedure. In Section \ref{sec:invertibility}, we develop tests for invertibility. In Section \ref{sec:validity}, we develop a necessary and sufficient condition for the moment condition in the second step to hold for the true production function. We show that ensuring that any instrument used in the second step is appropriately included in the regression in the first step either eliminates or mitigates the bias that arises if invertibility fails. In Section \ref{sec:modification_orthogonalization}, we modify the moment condition in the second step to endow it with Neyman orthogonality. In Section \ref{sec:montecarlo}, we conduct a Monte Carlo exercise to illustrate what goes wrong in the OP/LP/ACF procedure if invertibility fails and what can still be done. In Section \ref{sec:sensitivity}, we provide a diagnostic to assess the sensitivity of the estimates to the size of the prediction error. We conclude in Section \ref{sec:conclusion}.

\section{Setup and OP/LP/ACF procedure\label{sec:model}}

Firm $i$ in period $t$ produces output $Q_{it}$ with inputs $K_{it}$ and $V_{it}$ according to the production function
\begin{equation}\label{eq:prodfun}
q_{it}=f(k_{it},v_{it})+\omega_{it}+\varepsilon_{it},
\end{equation}
where lower case letters denote logs. Capital $k_{it}$ is a predetermined input that is chosen in period $t-1$ whereas $v_{it}$ is freely variable and decided on in period $t$ after the firm observes its productivity $\omega_{it}$.\footnote{While $k_{it}$ and $v_{it}$ may be vectors, we think of them as scalars for simplicity. The variable input $v_{it}$ may accordingly be interpreted as a composite of labor and materials such as cost of goods sold.} Productivity follows a first-order Markov process with law of motion
\begin{equation}\label{eq:law}
\omega_{it}=\expt\left[\omega_{it}|\omega_{it-1}\right]+\xi_{it}=g(\omega_{it-1})+\xi_{it},
\end{equation}
where the productivity innovation $\xi_{it}$ is by construction mean independent of lagged productivity $\omega_{it-1}$ and further assumed to be mean independent of any variable included in the firm's information set in period $t-1$. The disturbance $\varepsilon_{it}$ sits between the firm's output $q_{it}$ as recorded in the data and the output $q^*_{it}=q_{it}-\varepsilon_{it}=f(k_{it},v_{it})+\omega_{it}$ that the firm planned on when it decided on the variable input $v_{it}$. It can be interpreted as measurement error or as an unanticipated shock to output (OP, pp. 1273--1274) and is assumed to be mean independent of the inputs and other included variables as formalized below.\footnote{\label{fn:component}\citeasnoun{MUND:65} refer to $\omega_{it}$ and $\varepsilon_{it}$ as the transmitted, respectively, untransmitted component of productivity. The untransmitted component may include machine breakdowns, labor actions, supply chain disruptions, and power outages that are not anticipated by the firm. The interpretation as measurement error accommodates serial correlation in the disturbance $\varepsilon_{it}$.}  While the econometrician observes actual output $q_{it}$ and the inputs $k_{it}$ and $v_{it}$, productivity $\omega_{it}$, the disturbance $\varepsilon_{it}$, and planned output $q^*_{it}$ remain unobserved.

\paragraph{Invertibility.}

The literature following OP relies on invertibility. Invertibility assumes that there exists a function $\omega_{it}=h(x_{it})$ that maps observables $x_{it}=\left(k_{it},v_{it},\ldots\right)$ into productivity $\omega_{it}$. This is equivalently to
\begin{equation*}
\expt\left[\omega_{it}|x_{it}\right]=\omega_{it}.
\end{equation*}
The observables included in $x_{it}$ depend on the decision of the firm that is being inverted.\footnote{OP invert the firm's demand for investment whereas LP and ACF invert its demand for materials.} We remain agnostic about which variables are included in $x_{it}$, aside from imposing, without loss of generality, that the inputs $k_{it}$ and $v_{it}$ are included.

\paragraph{OP/LP/ACF procedure.}

Estimation proceeds in two steps. Step 1 assumes $\expt\left[\varepsilon_{it}|x_{it}\right]=0$ and flexibly or nonparametrically estimates the conditional expectation
\begin{equation}\label{eq:step1}
\expt\left[q_{it}|x_{it}\right]=\expt\left[\left.f(k_{it},v_{it})+\omega_{it}+\varepsilon_{it}\right|x_{it}\right]
=f(k_{it},v_{it})+\expt\left[\omega_{it}|x_{it}\right]=f(k_{it},v_{it})+\omega_{it} = q^*_{it},
\end{equation}
where the first equality uses equation \eqref{eq:prodfun}, the second-to-last equality uses invertibility, and the last equality uses the definition of planned output $q^*_{it}$. Because $\expt\left[q_{it}|x_{it}\right]=q^*_{it}$, step 1 separates actual output $q_{it}$ into planned output $q^*_{it}$ and the disturbance $\varepsilon_{it}=q_{it}-q^*_{it}$.

Step 2 of the OP/LP/ACF procedure assumes $\expt\left[\left.\xi_{it}+\varepsilon_{it}\right|z_{it}\right]=0$ for instruments $z_{it}=\left(k_{it},k_{it-1},v_{it-1},\ldots\right)$ and estimates the parameters $\theta=(\theta_f,\theta_g)$ in the production function $f$ and the law of motion $g$ by GMM.\footnote{For notational convenience we suppress $\theta$ in much of what follows.} Capital $k_{it}$ is a valid instrument because it is a predetermined input, and the lagged inputs $k_{it-1}$ and $v_{it-1}$ are valid instruments because the productivity innovation $\xi_{it}$ is mean independent of any variable included in the firm's information set in period $t-1$. We remain agnostic regarding additional instruments in $z_{it}$.

Using equations \eqref{eq:prodfun} and \eqref{eq:law} and the definition of planned output $q^*_{it}$, the assumption $\expt\left[\left.\xi_{it}+\varepsilon_{it}\right|z_{it}\right]=0$ implies
	\begin{equation}\label{eq:true_step2}
		\expt[q_{it} - f(k_{it}, v_{it}) - g(\omega_{it-1} )|z_{it}] =
		\expt[q_{it} - f(k_{it}, v_{it}) - g(q^*_{it-1} - f(k_{it-1}, v_{it-1}))|z_{it}] = 0.
	\end{equation}
Moment condition \eqref{eq:true_step2} is infeasible for estimation because lagged planned output $q^*_{it-1}$ is unobserved. Substituting $\expt[q_{it-1}|x_{it-1}]$ from step 1 for $q^*_{it-1}$, step 2 therefore proceeds by estimating the parameters $\theta$ from the moment condition
\begin{equation}
\expt\left[q_{it}-f(k_{it},v_{it})-g\left(\expt\left[q_{it-1}|x_{it-1}\right]-f(k_{it-1},v_{it-1})\right)|z_{it}\right] = 0. \label{eq:step2}
\end{equation}

Throughout the remainder of the paper, we follow OP, LP, and ACF and maintain $\expt\left[\varepsilon_{it}|x_{it}\right]=0$ for observables $x_{it}=\left(k_{it},v_{it},\ldots\right)$ and $\expt\left[\left.\xi_{it}+\varepsilon_{it}\right|z_{it}\right]=0$ for instruments $z_{it}=\left(k_{it},k_{it-1},v_{it-1},\ldots\right)$. To avoid cumbersome notation, we adopt the convention that all equalities involving random variables and conditional expectations are understood to hold almost surely.

\section{Tests for invertibility\label{sec:invertibility}}

If invertibility fails and $\expt\left[\omega_{it}|x_{it}\right]\neq\omega_{it}$, then productivity $\omega_{it}$ becomes a hidden state in a Markov model. The literature on Kalman and particle filters in dynamic systems shows that the best guess for the unobservables uses the entire history of the observables as opposed to just their current value. Our tests for invertibility are based on this intuition.

Our first proposition shows that invertibility implies a testable mean-independence restriction:
\begin{proposition}\label{prop:test1}
If $\expt\left[\omega_{it}|x_{it}\right]=\omega_{it}$ and $\expt\left[\varepsilon_{it}|x_{it},x_{it-1}\right]=\expt\left[\varepsilon_{it}|x_{it}\right]$, then $\expt\left[q_{it}|x_{it},x_{it-1}\right]=\expt\left[q_{it}|x_{it}\right]$.
\end{proposition}

While step 1 of the OP/LP/ACF procedure assumes $\expt\left[\varepsilon_{it}|x_{it}\right]=0$ for observables $x_{it}=\left(k_{it},v_{it},\ldots\right)$ and step 2 assumes $\expt\left[\xi_{it}+\varepsilon_{it}|z_{it}\right]=0$ for instruments $z_{it}=\left(k_{it},k_{it-1},v_{it-1},\ldots\right)$, this does not quite imply $\expt\left[\varepsilon_{it}|x_{it},x_{it-1}\right]=\expt\left[\varepsilon_{it}|x_{it}\right]$.
However, the interpretation of the disturbance $\varepsilon_{it}$ as an unanticipated shock to output that is outside the firm's information set in period $t$ implies $\expt\left[\varepsilon_{it}|x_{it},x_{it-1}\right]=\expt\left[\varepsilon_{it}|x_{it}\right] = 0$.\footnote{Under the interpretation of the disturbance $\varepsilon_{it}$ as measurement error, the assumption $\expt\left[\varepsilon_{it}|x_{it},x_{it-1}\right]=\expt\left[\varepsilon_{it}|x_{it}\right]$ holds if $x_{it-1}$ does not have predictive power for $\varepsilon_{it}$ beyond $x_{it}$.} The proof of Proposition \ref{prop:test1} is straightforward:
\begin{proof}
If $\expt\left[\omega_{it}|x_{it}\right]=\omega_{it}$, then
\begin{equation*}
q_{it}=f(k_{it},v_{it})+\omega_{it}+\varepsilon_{it}=f(k_{it},v_{it})+\expt\left[\omega_{it}|x_{it}\right]+\varepsilon_{it}.
\end{equation*}
$\expt\left[\varepsilon_{it}|x_{it},x_{it-1}\right]=\expt\left[\varepsilon_{it}|x_{it}\right]$ therefore implies $\expt\left[q_{it}|x_{it},x_{it-1}\right]=\expt\left[q_{it}|x_{it}\right]$.
\end{proof}

Our second proposition shows that invertibility implies a testable conditional-independence restriction:
\begin{proposition}\label{prop:test2}
If $\expt\left[\omega_{it}|x_{it}\right]=\omega_{it}$ and $\left.\varepsilon_{it} \independent x_{it-1}\right|x_{it}$, then $q_{it} \independent x_{it-1}|x_{it}$.
\end{proposition}

Going beyond the first-order implication of $\expt\left[\omega_{it}|x_{it}\right]=\omega_{it}$ in Proposition \ref{prop:test1} calls for a stronger assumption on the disturbance $\varepsilon_{it}$. The assumption $\left.\varepsilon_{it} \independent x_{it-1}\right|x_{it}$ in Proposition \ref{prop:test2} means that any possible dependence between $\varepsilon_{it}$ and $x_{it-1}$ is through $x_{it}$. It is implied by $\varepsilon_{it} \independent (x_{it},x_{it-1})$ and implies $\expt\left[\varepsilon_{it}|x_{it},x_{it-1}\right]=\expt\left[\varepsilon_{it}|x_{it}\right]$. The proof of Proposition \ref{prop:test2} parallels the proof of Proposition \ref{prop:test1} and is therefore omitted.

We implement our tests for invertibility on an unbalanced panel of Spanish manufacturing firms from 1990 to 2006 (Encuesta Sobre Estrategias Empresariales) and a balanced panel of US manufacturing industries from 1958 to 2018 (NBER-CES). In both cases, we reject invertibility by a wide margin. We provide further details on the data and the tests in Appendix \ref{app:invertibility}. In the next section, we turn to the consequences of a failure of invertibility for production function estimation.

\section{Failure of invertibility and validity of OP/LP/ACF moment condition\label{sec:validity}}

If invertibility fails and $\expt\left[\omega_{it}|x_{it}\right]\neq\omega_{it}$, then a prediction error $\zeta_{it}=\omega_{it}-\expt\left[\omega_{it}|x_{it}\right]$ arises in step 1 of the OP/LP/ACF procedure that is generally non-zero. By construction, the prediction error satisfies $\expt\left[\zeta_{it}|x_{it}\right]=0$.

To see how the presence of the prediction error $\zeta_{it}$ affects the OP/LP/ACF procedure, recall that the conditional expectation estimated in step 1 is
\begin{equation}\label{eq:step1new}
	\expt\left[q_{it}|x_{it}\right]=f(k_{it},v_{it})+\expt\left[\omega_{it}|x_{it}\right] =f(k_{it},v_{it})+\omega_{it}-\zeta_{it} = q^*_{it} - \zeta_{it},
\end{equation}
where the first equality uses equation \eqref{eq:prodfun} and $\expt[\varepsilon_{it}|x_{it}] = 0$, the second equality uses the definition of the prediction error $\zeta_{it}$, and the last equality uses the definition of planned output $q^*_{it}$. Equation \eqref{eq:step1new} provides an alternative interpretation of the prediction error as the difference between planned output and its prediction in step 1: $\zeta_{it} = q^*_{it} - \expt[q_{it}|x_{it}]$.

Substituting $\expt[q_{it-1}|x_{it-1}]$ from equation \eqref{eq:step1new}, the left-hand side of moment condition \eqref{eq:step2} in step 2 of the OP/LP/ACF procedure becomes
\begin{equation}\label{eq:step2_with_prediction_error}
\expt\left[q_{it} -f(k_{it},v_{it})-g\left(q^*_{it-1} - \zeta_{it-1} -f(k_{it-1},v_{it-1})\right)|z_{it}\right].
\end{equation}
Because of the lagged prediction error $\zeta_{it-1}$, conditional moment \eqref{eq:step2_with_prediction_error} is generally non-zero. To see this more clearly, consider the special case of an $AR(1)$ process for productivity. If $g(\omega_{it-1})=\rho\omega_{it-1}$, then conditional moment \eqref{eq:step2_with_prediction_error} evaluates to
\begin{equation}\label{eq:step2ar1}
\expt\left[\xi_{it}+\varepsilon_{it}+\rho\zeta_{it-1}|z_{it}\right]=\rho \expt\left[\zeta_{it-1}|z_{it}\right],
\end{equation}
where the equality uses $\expt\left[\xi_{it}+\varepsilon_{it}|z_{it}\right]=0$. 

To further see how this affects the GMM estimation, recall that $x_{it-1}=\left(k_{it-1},v_{it-1},\ldots\right)$ and $z_{it}=\left(k_{it},k_{it-1},v_{it-1},\ldots\right)$. The lagged inputs $k_{it-1}$ and $v_{it-1}$ hence remain valid instruments because $\expt\left[\zeta_{it-1}|x_{it-1}\right]=0$ by construction. However, capital $k_{it}$ may no longer be a valid instrument as $\expt\left[\zeta_{it-1}|k_{it}\right]\neq 0$, as noted by \citeasnoun{DORA:21}. This is because the firm chooses $k_{it}$ in period $t-1$ with knowledge of $\omega_{it-1}$. What remains of lagged productivity after controlling for the lagged observables $x_{it-1}$ may therefore be correlated with $k_{it}$. This results in biased estimates.

The above discussion foreshadows part of the solution to the problem: adding the lead of capital $k_{it+1}$ to the observables $x_{it}$ in step 1 of the OP/LP/ACF procedure ensures that $\expt\left[\zeta_{it-1}|k_{it}\right]=0$ and hence the validity of capital $k_{it}$ as an instrument in step 2. This simple expedient amounts to rethinking how the OP/LP/ACF procedure is implemented. The remaining question is to what extent it suffices beyond the special case of an $AR(1)$ process for productivity.

The following theorem provides a complete answer to this question in the form of a necessary and sufficient condition.
\begin{theorem}\label{thm:main}
	Moment condition \eqref{eq:step2} in step 2 of the OP/LP/ACF procedure holds for the true production function $f^0$ and some law of motion $\tilde g$ if and only if
\begin{equation}\label{eq:iff_proxy_variable}
\expt\left[g^0(\omega_{it-1})|z_{it}\right]=\expt\left[\left.\tilde g\left(\expt\left[\omega_{it-1}|x_{it-1}\right]\right)\right|z_{it}\right],
\end{equation}
where $g^0$ is the true law of motion.
\end{theorem}

We prove Theorem \ref{thm:main} before parsing it:
\begin{proof}
Recall that the true production function $f^0$ and the true law of motion $g^0$ satisfy moment condition \eqref{eq:true_step2}.

\noindent {\em ``Only if'' part:} Suppose moment condition \eqref{eq:step2} holds for the true production $f^0$ and some law of motion $\tilde g$. Then we have
\begin{gather*}
0=\expt\left[\left.q_{it}-f^0(k_{it},v_{it})-\tilde g\left(\expt\left[q_{it-1}|x_{it-1}\right]-f^0(k_{it-1},v_{it-1})\right)\right|z_{it}\right] \\
=\expt\left[\left.q_{it}-f^0(k_{it},v_{it})-\tilde g\left(\expt\left[\omega_{it-1}|x_{it-1}\right]\right)\right|z_{it}\right],
\end{gather*}
where the last equality uses equation \eqref{eq:step1new}. Subtracting from equation \eqref{eq:true_step2} implies condition \eqref{eq:iff_proxy_variable}.

\noindent {\em ``If'' part:} Suppose condition \eqref{eq:iff_proxy_variable} holds. Substituting into equation \eqref{eq:true_step2}, we have
\begin{gather*}
0=\expt\left[\left.q_{it}-f^0(k_{it},v_{it})-\tilde g\left(\expt\left[\omega_{it-1}|x_{it-1}\right]\right)\right|z_{it}\right] \\
=\expt\left[\left.q_{it}-f^0(k_{it},v_{it})-\tilde g\left(\expt\left[q_{it-1}|x_{it-1}\right]-f^0(k_{it-1},v_{it-1})\right)\right|z_{it}\right],
\end{gather*}
where the last equality uses equation \eqref{eq:step1new}. Hence, moment condition \eqref{eq:step2} holds for $f^0$ and $\tilde g$.
\end{proof}

Theorem \ref{thm:main} covers the invertibility assumption at the heart of the literature following OP as a special case. If $\expt\left[\omega_{it}|x_{it}\right]=\omega_{it}$, then condition \eqref{eq:iff_proxy_variable} becomes $\expt\left[g^0(\omega_{it-1})|z_{it}\right]=\expt\left[\tilde g\left(\omega_{it-1}\right)|z_{it}\right]$ and is therefore satisfied with $\tilde g=g^0$. Importantly, however, condition \eqref{eq:iff_proxy_variable} can be satisfied even if invertibility fails and $\expt\left[\omega_{it}|x_{it}\right]\neq \omega_{it}$.

To understand condition \eqref{eq:iff_proxy_variable}, note that it becomes $\expt\left[g^0(\omega_{it-1})|z_{it}\right]=\tilde g\left(\expt\left[\omega_{it-1}|x_{it-1}\right]\right)$ if $x_{it-1}\subsetneq z_{it}$. This is difficult to satisfy as the left-hand side is a function of $z_{it}$ whereas the right-hand side is a function of $x_{it-1}\subsetneq z_{it}$.

We therefore set $x_{it-1}\supseteq z_{it}$ in what follows. This amounts to rethinking how the OP/LP/ACF procedure is implemented, as foreshadowed by the special case of an $AR(1)$ process for productivity. In particular, because $k_{it}$ is included in $z_{it}$ in step 2, we must include $k_{it+1}$ in $x_{it}$ in step 1.

With this choice of $x_{it-1}\supseteq z_{it}$, we further examine condition \eqref{eq:iff_proxy_variable} through a series of examples:

\begin{example}[linear law of motion, $AR(1)$ process]\label{ex:linear}
If $g^0(\omega_{it-1})=\rho\omega_{it-1}$, then $\expt\left[\rho\omega_{it-1}|z_{it}\right]=\expt\left[\rho \expt\left[\omega_{it-1}|x_{it-1}\right]|z_{it}\right]$ by the law of iterated expectations. Condition (\ref{eq:iff_proxy_variable}) is therefore satisfied with $\tilde g(\expt[\omega_{it-1}|x_{it-1}])=g^0(\expt[\omega_{it-1}|x_{it-1}])$.
\end{example}
\noindent Remarkably, moment condition \eqref{eq:step2} in step 2 of the OP/LP/ACF procedure holds in Example \ref{ex:linear} irrespective of the quality of the estimate of $\expt\left[q_{it}|x_{it}\right]$ in step 1.

\begin{example}[quadratic law of motion]\label{ex:quadratic}
If $g^0(\omega_{it-1})=\rho_1\omega_{it-1}+\rho_2\omega_{it-1}^2$ and $\var\left(\zeta_{it-1}|z_{it}\right)=\sigma^2$, then
\begin{equation*}
\expt\left[\left.\rho_1\omega_{it-1}+\rho_2\omega_{it-1}^2\right|z\right]=\expt\left[\left.\rho_1\expt\left[\omega_{it-1}|x_{it-1}\right]+\rho_2\expt\left[\omega_{it-1}|x_{it-1}\right]^2+\rho_2\zeta_{it-1}^2\right|z_{it}\right].
\end{equation*}
Condition \eqref{eq:iff_proxy_variable} is therefore satisfied with $\tilde g(\expt[\omega_{it-1}|x_{it-1}])=g^0(\expt[\omega_{it-1}|x_{it-1}])+\rho_2\sigma^2$.\footnote{We can relax $\var\left(\zeta_{it-1}|z_{it}\right)=\sigma^2$ to $\var\left(\zeta_{it-1}|z_{it}\right) = h(\expt[\omega_{it-1}|x_{it-1}])$ for some function $h$. In this case, condition \eqref{eq:iff_proxy_variable} holds with $\tilde g(\expt[\omega_{it-1}|x_{it-1}])=g^0(\expt[\omega_{it-1}|x_{it-1}])+\rho_2 h(\expt[\omega_{it-1}|x_{it-1}])$. }
\end{example}

Example \ref{ex:quadratic} is less restrictive on the law of motion but more restrictive on the lagged prediction error than Example \ref{ex:linear}. The next example pushes this tradeoff further:
\begin{example}[analytic law of motion]\label{ex:analytic}
If the Taylor series of $g^0$ around $\expt[\omega_{it-1}|x_{it-1}]$ converges absolutely so that $g^0(\omega_{it-1})=\sum_{j=0}^\infty \frac{g^{0,(j)}(\expt[\omega_{it-1}|x_{it-1}])}{j!}\zeta^j_{it-1}$, where $g^{0, (j)}$ is the $j$th derivative of $g^0$, then
\begin{gather*}
\expt\left[\left.\sum_{j=0}^\infty \frac{g^{0,(j)}(\expt[\omega_{it-1}|x_{it-1}])}{j!}\zeta^j_{it-1}\right|z_{it}\right]
=\expt\left[\left.\expt\left[\left.\sum_{j=0}^\infty\frac{g^{0,(j)}(\expt[\omega_{it-1}|x_{it-1}])}{j!}\zeta^j_{it-1}\right|x_{it-1}\right]\right|z_{it}\right] \\
=\expt\left[\left.\sum_{j=0}^\infty\frac{g^{0,(j)}(\expt[\omega_{it-1}|x_{it-1}])}{j!}\expt[\zeta^j_{it-1}|x_{it-1}]\right|z_{it}\right].
\end{gather*}
If $\expt[\zeta^j_{it-1}|x_{it-1}]=\alpha_j$ for some constant $\alpha_j$ for all $j$, then condition \eqref{eq:iff_proxy_variable} is satisfied in this example with $\tilde g(\expt[\omega_{it-1}|x_{it-1}])=\sum_{j=0}^\infty\frac{g^{0,(j)}(\expt[\omega_{it-1}|x_{it-1}])}{j!}\alpha_j$.\footnote{Condition \eqref{eq:iff_proxy_variable} remains valid if we relax the condition that $\expt[\zeta^j_{it-1}|x_{it-1}]=\alpha_j$ for some constant $\alpha_j$ for all  $j$ to $\left.\zeta_{it-1}\independent z_{it}\right| \expt[\omega_{it-1}|x_{it-1}]$.}
\end{example}

Beyond these examples, condition \eqref{eq:iff_proxy_variable} may be violated even if $x_{it-1}\supseteq z_{it}$.\footnote{Testing condition \eqref{eq:iff_proxy_variable} is empirically challenging. A necessary condition for it to be satisfied is that moment condition \eqref{eq:step2} in step 2 of the OP/LP/ACF procedure holds at the estimated production function and law of motion. This can be tested with an overidentification test that corrects for the plug-in nature of the OP/LP/ACF procedure and clustering in the data.} While this results in biased estimates, setting $x_{it-1}\supseteq z_{it}$ mitigates the bias. If $x_{it-1}\supseteq z_{it}$, then we have $\expt\left[\omega_{it-1}|z_{it}\right]=\expt\left[\expt\left[\omega_{it-1}|x_{it-1}\right]|z_{it}\right]$ by the law of iterated expectations and hence $\expt\left[\zeta_{it-1}|z_{it}\right]$ $=$ $\expt\left[\left.\omega_{it-1}-\expt\left[\omega_{it-1}|x_{it-1}\right]\right|z_{it}\right]$ $=0$. This mean independence of the lagged prediction error $\zeta_{it-1}$ from the instruments $z_{it}$ intuitively eliminates a first-order source of bias. The following theorem formalizes this intuition:

\begin{theorem}\label{thm:bias_reduction}
If $x_{it-1}\supseteq z_{it}$ and $g$ is differentiable, then moment condition \eqref{eq:step2} in step 2 of the OP/LP/ACF procedure implicitly incorporates a first-order bias correction:
	\begin{multline}\label{eq:first_order_correction}
\expt\left[q_{it}-f(k_{it},v_{it})-g\left(\expt\left[q_{it-1}|x_{it-1}\right]-f(k_{it-1},v_{it-1})\right)|z_{it}\right] = \expt\left[q_{it}-f(k_{it},v_{it})  \right.\\
 -  g\left(q^*_{it-1} - \zeta_{it-1} -f(k_{it-1},v_{it-1}) \right) - \left. g^\prime\left(q^*_{it-1} - \zeta_{it-1} -f(k_{it-1},v_{it-1}) \right)\zeta_{it-1} \big| z_{it}\right],
	\end{multline}
where $g^{\prime}$ is the first derivative of $g$.
 \end{theorem}
\noindent We note that equation \eqref{eq:first_order_correction} holds for any production function $f$ and any law of motion $g$ for which the conditional moments exist. The proof of Theorem \ref{thm:bias_reduction} is in Appendix \ref{app:validity}.

To appreciate Theorem \ref{thm:bias_reduction}, recall that moment condition \eqref{eq:true_step2} is valid even if invertibility fails. Because lagged planned output $q^*_{it-1}$ is unobserved, moment condition \eqref{eq:true_step2} is infeasible for estimation. Step 2 of the OP/LP/ACF procedure therefore proceeds by substituting $\expt\left[q_{it-1}|x_{it-1}\right]$ from step 1 for $q^*_{it-1}$. If invertibility fails, then this substitution introduces the lagged prediction error $\zeta_{it-1}=q^*_{it-1}-\expt\left[q_{it-1}|x_{it-1}\right]$ into moment condition \eqref{eq:step2} (see conditional moment \eqref{eq:step2_with_prediction_error}). Theorem \ref{thm:bias_reduction} shows that setting $x_{it-1}\supseteq z_{it}$ adjusts moment condition \eqref{eq:step2} toward the valid but infeasible moment condition \eqref{eq:true_step2}.

Indeed, comparing the term $g\left(\expt\left[q_{it-1}|x_{it-1}\right]-f(k_{it-1},v_{it-1})\right)$ in moment condition \eqref{eq:step2} to the term $g\left(q^*_{it-1}-f(k_{it-1},v_{it-1})\right)$ in moment condition \eqref{eq:true_step2} yields
\begin{eqnarray}
 && g\left(\expt\left[q_{it-1}|x_{it-1}\right]-f(k_{it-1},v_{it-1}) \right) - g\left(q^*_{it-1}-f(k_{it-1},v_{it-1})\right)\nonumber \\
 &=& g\left(q^*_{it-1} - \zeta_{it-1}-f(k_{it-1},v_{it-1}) \right) - g\left(q^*_{it-1}-f(k_{it-1},v_{it-1})\right)\nonumber\\
 &=& -g^\prime\left(q^*_{it-1} - \zeta_{it-1}-f(k_{it-1},v_{it-1}) \right) \zeta_{it-1} + o(\zeta_{it-1}), \label{eq:first_order_expansion}
 \end{eqnarray}
where the last equality uses a first-order Taylor expansion of $g$ around $q^*_{it-1} - \zeta_{it-1} -f(k_{it-1},v_{it-1})$ and $o(\zeta_{it-1})$ denotes the higher-order remainder. Equation \eqref{eq:first_order_correction} in Theorem \ref{thm:bias_reduction} shows that if $x_{it-1}\supseteq z_{it}$, then moment condition \eqref{eq:step2} implicitly accounts for the term $-g^\prime\left(q^*_{it-1} - \zeta_{it-1}-f(k_{it-1},v_{it-1}) \right) \zeta_{it-1}$. Moment condition \eqref{eq:step2} thus differs from the valid but infeasible moment condition \eqref{eq:true_step2} at most by higher-order terms. Put differently, setting $x_{it-1}\supseteq z_{it}$ provides a first-order bias correction in step 2 of the OP/LP/ACF procedure.

Because setting $x_{it-1}\supseteq z_{it}$ provides a first-order bias correction in step 2 of the OP/LP/ACF procedure, a violation of condition \eqref{eq:iff_proxy_variable} must be of second order relative to size of the prediction error. The following theorem formalizes this intuition:
\begin{theorem}\label{thm:bound}
Let $\mathcal{G}$ be a set of functions that includes the true law of motion $g^0$. If $x_{it-1}\supseteq z_{it}$ and $g^{0}$ is twice continuously differentiable, then
\begin{equation*}
	\inf_{\tilde g\in \mathcal{G}}\norm{\expt\left[g^0(\omega_{it-1})|z_{it}\right]-\expt\left[\tilde g\left(\expt\left[\omega_{it-1}|x_{it-1}\right]\right)|z_{it}\right]}_{L,1}\leq \tau\var\left(\zeta_{it-1}\right),
\end{equation*}
where $\norm{\cdot}_{L, 1}$ is the $L_1$ norm, $\tau=\sup_{\omega_{it-1}}\left|g^{0\prime\prime}(\omega_{it-1})\right|$, and $g^{0\prime\prime}$ is the second derivative of $g^0$.
\end{theorem}
\noindent The proof of Theorem \ref{thm:bound} is in Appendix \ref{app:validity}.

Theorem \ref{thm:bound} ties the violation of condition \eqref{eq:iff_proxy_variable} to the variance of the lagged prediction error $\zeta_{it-1}$. Because $\var(\zeta_{it-1}) = \expt[ (\omega_{it-1} - \expt[\omega_{it-1}|x_{it-1}])^2]$, this variance gets smaller as the observables $x_{it-1}$ get richer even beyond $z_{it}$. Theorem \ref{thm:bound} therefore suggests to take a kitchen sink approach to the regression in step 1 of the OP/LP/ACF procedure and add as many relevant covariates as possible. Adding covariates that speak to demand conditions and firm conduct may be especially helpful.\footnote{The recent literature in fact proceeds along this line: ``The materials demand function in our setting will take as arguments
all state variables of the firm (\ldots), including productivity,
and all additional variables that affect a firm's demand for materials. These
include firm location (\ldots), output prices (\ldots), product dummies (\ldots), market
shares (\ldots), input prices (\ldots), the export status of a firm (\ldots), and the
input (\ldots) and output tariffs (\ldots) that the firm faces on the product it
produces'' \cite[p. 466]{DELO:15}. Whether adding covariates restores invertibility can be assessed using the tests in Section \ref{sec:invertibility}.}

Adding covariates may make the assumption $\expt\left[\varepsilon_{it}|x_{it}\right]=0$ that underpins step 1 of the OP/LP/ACF procedure more demanding. We advocate replacing $x_{it}=\left(k_{it},v_{it},\ldots\right)$ in the original procedure by $x_{it}=\left(k_{it+1},k_{it},v_{it},\ldots\right)$. In the original procedure, the assumption $\expt\left[\varepsilon_{it}|x_{it}\right]=0$ is justified by stipulating that $\varepsilon_{it}$ is outside the firm's information set in period $t$. This justification extends to our approach as the firm is assumed to choose $k_{it+1}$ in period $t$. However, adding leads of other covariates to $x_{it}$ may be more challenging if the firm eventually learns $\varepsilon_{it}$. To guard against this possibility, we recommend adding current values or lags of other covariates to $x_{it}$.

Finally, we caution that adding covariates introduces a curse of dimensionality. The resulting larger mean squared error of the estimate of $\expt[q_{it-1}| x_{it-1}]$ may counterbalance the benefit of the smaller $\var(\zeta_{it-1})$ in finite samples. In the next section, we propose a modification of the OP/LP/ACF procedure that addresses this problem.

\paragraph{Summary.}

Our analysis calls for rethinking how the OP/LP/ACF procedure is implemented. Theorem \ref{thm:main} makes a strong case for ensuring $x_{it-1}\supseteq z_{it}$. For a linear law of motion, the moment condition in step 2 of the OP/LP/ACF procedure holds for the true production function $f^0$ and the true law of motion $g^0$. For a nonlinear law of motion, recovering the true production function $f^0$ requires condition \eqref{eq:iff_proxy_variable} to be satisfied. A violation of condition \eqref{eq:iff_proxy_variable} results in biased estimates. In this case, Theorem \ref{thm:bias_reduction} shows that setting $x_{it-1}\supseteq z_{it}$ provides a first-order bias correction.

Theorem \ref{thm:main} further suggests to be flexible in modeling the law of motion in step 2 as condition \eqref{eq:iff_proxy_variable} is easier to satisfy if $\tilde g$ can come from a larger set $\mathcal{G}$. Theorem \ref{thm:bound} finally suggests to take a kitchen sink approach to step 1.

\paragraph{Related literature.}

An alternative to relaxing the invertibility assumption is to forgo step 1 of the OP/LP/ACF procedure. Instead of substituting $\expt\left[q_{it-1}|x_{it-1}\right]$ for $q^*_{it-1}$ in moment condition \eqref{eq:true_step2}, the dynamic panel approach to production function estimation pioneered by \citeasnoun{BLUN:00} substitutes $q_{it-1}$ and imposes the linear law of motion $g(\omega_{it-1}) = \rho \omega_{it-1}$ to obtain
\begin{equation*}
	\expt[q_{it} - f(k_{it}, v_{it}) - g(q_{it-1} - f(k_{it-1}, v_{it-1}))|z_{it}] =	\expt[q_{it} - f(k_{it}, v_{it}) - g(q^*_{it-1} - f(k_{it-1}, v_{it-1})) - \rho \varepsilon_{it-1}|z_{it}].
\end{equation*}
Besides the assumption $\expt\left[\xi_{it}+\varepsilon_{it}|z_{it}\right]=0$ maintained in moment condition \eqref{eq:true_step2}, the dynamic panel approach requires assuming $\expt[\varepsilon_{it-1}|z_{it}] = 0$. In comparison, our Example \ref{ex:linear} requires setting $x_{it-1}\supseteq z_{it}$ and assuming $\expt[\varepsilon_{it-1}|x_{it-1}] = 0$ in step 1 of the OP/LP/ACF procedure. Either assumption can be justified by appealing to the firm's information set and the timing of decisions, as discussed above.

\citeasnoun{HU:20}, \citeasnoun{BRAN:20}, and \citeasnoun{POND:21} generalize the dynamic panel approach from a linear to a nonlinear law of motion.\footnote{The estimation strategy in \citeasnoun{HU:20} and \citeasnoun{BRAN:20} differs from their identification strategies. The latter builds on \citeasnoun{HU:08} and relies on the existence of two conditionally independent proxies for productivity and the invertibility of an integral operator. The implied semiparametric maximum likelihood estimator involves significant computational challenges that preclude its implementation.}  To illustrate the similarities and differences with our approach, consider the quadratic law of motion $g(\omega_{it-1}) = \rho_1 \omega_{it-1} + \rho_2 \omega^2_{it-1}$. Substituting $q_{it-1}$ for $q^*_{it-1}$ in moment condition \eqref{eq:true_step2} yields
\begin{multline*}
	\expt[q_{it} - f(k_{it}, v_{it}) - g(q_{it-1}  - f(k_{it-1}, v_{it-1})) |z_{it}]
=\expt[q_{it} - f(k_{it}, v_{it}) \\ - g(q^*_{it-1}  - f(k_{it-1}, v_{it-1}))
- \rho_1\varepsilon_{it-1} - 2\rho_2 \varepsilon_{it-1} (q^*_{it-1}  - f(k_{it-1}, v_{it-1})) -\rho_2\varepsilon_{it-1}^2|z_{it}].
\end{multline*}
The generalized dynamic panel approach therefore requires assuming $\expt[\varepsilon_{it-1}|z_{it}] = 0$, $\expt[\varepsilon_{it-1}\omega_{it-1}| z_{it}] = 0$, and $\var\left(\varepsilon_{it-1}|z_{it}\right) = \sigma^2$. In comparison, our Example \ref{ex:quadratic} requires setting $x_{it-1}\supseteq z_{it}$ and assuming $\expt[\varepsilon_{it-1}|x_{it-1}] = 0$ and $\var(\zeta_{it-1}|z_{it})=\sigma^2$. Both approaches are similar in that they impose exclusion restrictions on second-order moments of latent variables to ensure that the moment condition holds for the true production function.
%Note that stipulating that $\varepsilon_{it-1}$ is outside the firm's information set in period $t-1$ implies $\expt[\varepsilon_{it-1}|z_{it}] = \expt[\varepsilon_{it-1}|x_{it-1}] =0$ but not  $\var\left(\varepsilon_{it-1}|z_{it}\right) = \sigma^2$, and the definition of the prediction error implies $\expt[\zeta_{it-1}|x_{it-1}] = 0$ but not  $\var\left(\zeta_{it-1}|z_{it}\right) = \sigma^2$.

This similarity extends to the polynomial law of motion $g(\omega_{it-1}) = \sum_{j=1}^\infty \rho_{j}\omega^j_{it-1}$. In this case, the generalized dynamic panel approach requires assuming that $\expt[\varepsilon_{it-1}^j|\omega_{it-1}, z_{it}]$ is constant for all $j$. This restriction on the joint distribution of $\varepsilon_{it-1}$ and $\omega_{it-1}$ may be unpalatable if $\varepsilon_{it-1}$ and $\omega_{it-1}$ are interpreted as the untransmitted, respectively, transmitted component of productivity (see footnote \ref{fn:component}).\footnote{If $x_{it-1}\subseteq z_{it}$ and invertibility holds, then $\expt[\varepsilon_{it-1}^j|\omega_{it-1}, z_{it}]=\expt[\varepsilon_{it-1}^j|z_{it}]$. While this facilitates interpreting the restriction, we focus on the case where invertibility fails.} In comparison, our Example \ref{ex:analytic} requires assuming that $\expt[\zeta_{it-1}^j | x_{it-1}]$ is constant for all $j$.

The key difference is that our approach imposes restrictions on the lagged prediction error $\zeta_{it-1}$ rather than the lagged disturbance $\varepsilon_{it-1}$. Whereas the magnitude and properties of $\varepsilon_{it-1}$ are inherent in the data generating process, the advantage of our approach is that the magnitude of $\zeta_{it-1}$ can be reduced by adding covariates to the regression in step 1 of the OP/LP/ACF procedure. Proceeding to add covariates may shrink $\var(\zeta_{it-1}) = \expt[ (\omega_{it-1} - \expt[\omega_{it-1}|x_{it-1}])^2]$ towards zero if we approach invertibility in the limit. This makes our approach particularly attractive for datasets with a rich set of observables.

\section{Modification of OP/LP/ACF moment condition\label{sec:modification_orthogonalization}}

Step 1 of the OP/LP/ACF procedure estimates the conditional expectation $\expt[q_{it-1}|x_{it-1}]$. To understand the impact that a noisy estimate in step 1 has on the GMM estimator in step 2, we use $\hat{e}(x_{it-1})$ to denote the estimate of a nuisance parameter $e(x_{it-1})$ with true value $\expt[q_{it-1}|x_{it-1}]$. We remain agnostic about the estimation method used in step 1. 

Replacing $\expt[q_{it-1}|x_{it-1}]$ by $\hat{e}(x_{it-1})$ in moment condition \eqref{eq:step2} that underpins step 2 yields
\begin{equation}\label{eq:step2_feasible}
\expt\left[q_{it}-f(k_{it},v_{it})-g\left(\hat{e}(x_{it-1})-f(k_{it-1},v_{it-1})\right)|z_{it}\right] = 0. 
\end{equation}
Because $\hat{e}(x_{it-1})$ converges to $\expt[q_{it-1}|x_{it-1}]$ as the sample size increases, we write $\hat{e}(x_{it-1}) = \expt[q_{it-1}|x_{it-1}] + \lambda \delta(x_{it-1})$, where $\lambda$ converges to zero as sample size increases and $\delta(x_{it-1})$ is the direction of the finite-sample deviation of $\hat{e}(x_{it-1})$ from $\expt[q_{it-1}|x_{it-1}]$. The pathwise (or Gateaux) derivative of the left-hand side of moment condition \eqref{eq:step2_feasible} with respect to the direction $\delta(x_{it-1})$ at $\hat{e}(x_{it-1})=\expt[q_{it-1}|x_{it-1}]$ is
\begin{eqnarray*}
&&\pd{}{\lambda} \expt\left[q_{it}-f(k_{it},v_{it}) - g\left( \expt[q_{it-1}|x_{it-1}] + \lambda \delta(x_{it-1}) -f(k_{it-1},v_{it-1}) \right) \big| z_{it}\right] \Big |_{\lambda = 0}\\
&=& -\expt\left[g'\left( \expt[q_{it-1}|x_{it-1}] -f(k_{it-1},v_{it-1}) \right) \delta(x_{it-1}) \big| z_{it} \right].
\end{eqnarray*}
Because this derivative is generally non-zero, moment condition \eqref{eq:step2} is impacted by the deviation of $\hat{e}(x_{it-1})$ from $\expt[q_{it-1}|x_{it-1}]$ that arises from estimation noise in step 1. 

To eliminate this impact and its adverse consequences for the asymptotic distribution of the GMM estimator, we modify the moment condition in step 2 as follows:
\begin{multline}\label{eq:modified_step2}
	\expt\left[q_{it}-f(k_{it},v_{it}) - g\left(\hat{e}(x_{it-1}) -f(k_{it-1},v_{it-1}) \right) \right.\\
	- \left. g^\prime\left(\hat{e}(x_{it-1}) -f(k_{it-1},v_{it-1}) \right) (q_{it-1} - \hat{e}(x_{it-1})) \big| z_{it}\right] = 0.
\end{multline}
Our modification explicitly incorporates a feasible version of the first-order bias correction that is implicit in moment condition \eqref{eq:step2} if $x_{it-1} \supseteq z_{it}$ (see Theorem \ref{thm:bias_reduction}).

If $\hat{e}(x_{it-1})=\expt[q_{it-1}|x_{it-1}]$ and $x_{it-1} \supseteq z_{it}$, then we have $\expt[q_{it-1} - \hat{e}(x_{it-1})| z_{it}] = 0$ so that moment condition \eqref{eq:modified_step2} coincides with moment condition \eqref{eq:step2}. Our modification therefore retains the identification power of the original OP/LP/ACF procedure when $x_{it-1} \supseteq z_{it}$. 

However, unlike the original moment condition \eqref{eq:step2}, the modified moment condition \eqref{eq:modified_step2} is not sensitive to small deviations of $\hat e(x_{it-1})$ from $\expt[q_{it-1}|x_{it-1}]$, a property known as \emph{Neyman orthogonality} \cite{NEYM:59}:
\begin{theorem}[Neyman orthogonality]\label{thm:neyman_orth}
	Define the set of functions $\tilde{\mathcal{T}} =\{ \delta(x_{it-1})=e(x_{it-1}) - \expt[q_{it-1}| x_{it-1}]: e \in \mathcal{T} \}$, where $\mathcal{T}$ is the set of real-valued integrable functions of $x_{it-1}$. If $x_{it-1}\supseteq z_{it}$ and $g$ has a bounded and continuous derivative, then the pathwise (or Gateaux) derivative of the left-hand side of moment condition \eqref{eq:modified_step2} with respect to any direction $\delta\in\tilde{\mathcal{T}}$ at $\hat{e}(x_{it-1}) =\expt[q_{it-1}| x_{it-1}]$ is zero:
	\begin{multline*}
		\pd{}{\lambda} \expt\left[q_{it}-f(k_{it},v_{it}) - g\left( \expt[q_{it-1}|x_{it-1}] + \lambda \delta(x_{it-1}) -f(k_{it-1},v_{it-1}) \right) \right.\\
		- \left. g^\prime\left(\expt[q_{it-1}|x_{it-1}] + \lambda \delta(x_{it-1}) -f(k_{it-1},v_{it-1}) \right)(q_{it-1} - \expt[q_{it-1}|x_{it-1}] - \lambda \delta(x_{it-1})) \big| z_{it}\right] \Big |_{\lambda = 0} =  0.
	\end{multline*}
\end{theorem}
\noindent The proof of Theorem \ref{thm:neyman_orth} is in Appendix \ref{app:modification_orthogonalization}.\footnote{The assumption that $g$ has a bounded derivative simplifies stating Theorem \ref{thm:neyman_orth}. As shown in Appendix \ref{app:modification_orthogonalization}, it can be replaced by a much weaker assumption.}$^,$\footnote{We recommend ensuring $x_{it-1}\supseteq z_{it}$ when implementing our modification of the OP/LP/ACF procedure. Without $x_{it-1}\supseteq z_{it}$, the modified moment condition \eqref{eq:modified_step2} entails a first-order bias correction toward $\expt[q_{it-1}|x_{it-1}, z_{it}]$ instead of $\expt[q_{it-1}|x_{it-1}]$ and the expectation of the correction term conditional on $z_{it}$ is generally non-zero even if $\hat{e}(x_{it-1}) = \expt[q_{it-1}|x_{it-1}]$.}

Neyman orthogonality has received renewed attention in the double-debiased machine learning literature. Following the arguments in \citeasnoun{CHER:18} and \citeasnoun{CHER:22}, as long as $\hat{e}(x_{it-1})$ converges to $\expt[q_{it-1}|x_{it-1}]$ at a rate faster than $N^{-1/4}$, the GMM estimator based on moment condition \eqref{eq:modified_step2} is \emph{oracle efficient}. This means that its asymptotic distribution is {\em as if} the true value $\expt[q_{it-1}|x_{it-1}]$ of the nuisance parameter $e(x_{it-1})$ is known. The GMM estimator based on moment condition \eqref{eq:modified_step2} therefore achieves the efficient asymptotic distribution.

Besides efficiency, the modified moment condition \eqref{eq:modified_step2} enjoys further advantages over the original moment condition \eqref{eq:step2}. First, because Theorem \ref{thm:neyman_orth} allows for any integrable deviation, it accommodates a wide range of estimation methods in step 1 of the OP/LP/ACF procedure. This includes traditional nonparametric methods such as sieve, kernel, and LASSO estimators, which produce smooth estimates, as well as modern machine learning techniques such as neural networks and random forests, which may produce non-smooth estimates. Neyman orthogonality means that small deviations of $\hat{e}(x_{it-1})$ from $\expt[q_{it-1}|x_{it-1}]$ that can arise from estimation noise, functional approximation error in sieve estimation, local smoothing bias in kernel estimation, or regularization bias in LASSO and machine learning methods do not have a first-order impact on the GMM estimator in step 2. As long as $\hat{e}(x_{it-1})$ converges to $\expt[q_{it-1}|x_{it-1}]$ at a rate faster than $N^{-1/4}$, its asymptotic distribution is robust to the choice of estimation method and its implementation in step 1.\footnote{Achieving this rate of convergence requires properly implementing a nonparametric estimation method to avoid functional form bias and the validity of smoothness assumptions specific to the selected method. Because Neyman orthogonality is a local property, the first-order bias correction in moment condition \eqref{eq:modified_step2} loses its theoretical justification if $\hat{e}(x_{it-1})$ deviates substantially from $\expt[q_{it-1}|x_{it-1}]$.}

Second, moment condition \eqref{eq:modified_step2} eliminates the need to correct standard errors in step 2 to account for estimation noise in step 1 through techniques such as those in \citeasnoun{HAHN:18}. 

Unlike general approaches that generate Neyman orthogonality using influence functions or score projections as described in \citeasnoun{CHER:22}, moment condition \eqref{eq:modified_step2} is Neyman orthogonal by construction. By avoiding introducing high-dimensional nuisance parameters, our modification remains as straightforward to implement as the original OP/LP/ACF procedure.

Finally, we note that if OLS is used in step 1 of OP/LP/ACF procedure to estimate the conditional expectation $\expt[q_{it-1}|x_{it-1}]$, then explicitly including the first-order bias correction can be redundant in some cases. Converting moment condition \eqref{eq:modified_step2} into an unconditional moment condition yields
\begin{multline*}
 \expt\left[h(z_{it}) \Big(q_{it}-f(k_{it},v_{it}) - g\left(\hat{e}(x_{it-1}) -f(k_{it-1},v_{it-1}) \right) \right.\\
	 - \left. g^\prime\left(\hat{e}(x_{it-1}) -f(k_{it-1},v_{it-1}) \right) (q_{it-1} - \hat{e}(x_{it-1})) \Big)\right] = 0,
\end{multline*}
where $h(z_{it})$ is a vector function of the instruments $z_{it}$. The first-order bias correction is redundant in a finite sample if 
\begin{equation}\label{eq:finite_sample_corrective_term}
 \frac{1}{NT}\sum_{i}\sum_t h(z_{it})g^\prime\left(\hat{e}(x_{it-1}) -f(k_{it-1},v_{it-1}) \right) (q_{it-1} - \hat{e}(x_{it-1}))=0.
\end{equation}
Because $q_{it-1} - \hat{e}(x_{it-1})$ is the residual from the OLS estimation in step 1, equation \eqref{eq:finite_sample_corrective_term} holds if $h(z_{it})g^\prime(\hat{e}(x_{it-1}) -f(k_{it-1},v_{it-1}))$ can be expressed as a linear combination of the regressors used in step 1. This is the case if the law of motion $g$ is linear and the regressors include $h(z_{it})$.
More generally, if $h(z_{it})g^\prime(\hat{e}(x_{it-1}) -f(k_{it-1},v_{it-1}))$ is well approximated by a linear combination of the regressors used in step 1, then the left-hand side of equation \eqref{eq:finite_sample_corrective_term} is almost zero and we expect the GMM estimator in step 2 to attain near-oracle performance. While our Monte Carlo exercise in Section \ref{sec:montecarlo} illustrates that this can happen, we recommend basing the GMM estimator in step 2 on the modified moment condition \eqref{eq:modified_step2}, as this ensures oracle efficiency irrespective of whether equation \eqref{eq:finite_sample_corrective_term} holds or not. Moreover, if estimation methods other than OLS are used in step 1, then equation \eqref{eq:finite_sample_corrective_term} generally does not hold and the modified moment condition \eqref{eq:modified_step2} improves efficiency.

\section{Monte Carlo exercise\label{sec:montecarlo}}

\paragraph{Data generating process.}

We specify the CES production function
\begin{equation*}
f(k_{it},v_{it})=\frac{\nu}{\rho}\ln\left(\alpha\exp(\rho k_{it})+(1-\alpha)\exp(\rho v_{it})\right)
\end{equation*}
and the disturbance $\varepsilon_{it}\sim N\left(0,\sigma^2_\varepsilon\right)$, where $\alpha\in(0,1)$ is a distributional parameter, $\rho\leq 0$ and $\sigma=\frac{1}{1-\rho}$ is the elasticity of substitution, and $\nu>0$ is returns to scale. We further specify the law of motion
\begin{equation*}
g(\omega_{it-1})=\mu_\omega+\rho_\omega\left((1-\alpha_\omega)\omega_{it-1}+\frac{\alpha_\omega}{6}\ln(\ln(1+\exp(6\omega_{it-1})))\right)
\end{equation*}
and the productivity innovation $\xi_{it}\sim N\left(0,\sigma^2_\omega\right)$, where $\rho_\omega\in(0,1)$ and $\alpha_\omega\in[0,1]$. Our particular interest lies in comparing the Gaussian $AR(1)$ process for $\alpha_\omega=0$ with the nonlinear process for $\alpha_\omega=1$. To facilitate this comparison we calibrate $\mu_\omega$, $\rho_\omega$, and $\sigma^2_\omega$ to hold fixed $\expt\left[\omega_{it}\right]$, $\var\left(\omega_{it}\right)$, and $\corr\left(\omega_{it},\omega_{it-1}\right)$. Figure \ref{fig:productivity} illustrates the resulting distribution of productivity and overlays the law of motion for the two processes.

\begin{figure}
\centerline{\includegraphics[scale=0.5,page=1,trim={0 4in 0 0},clip]{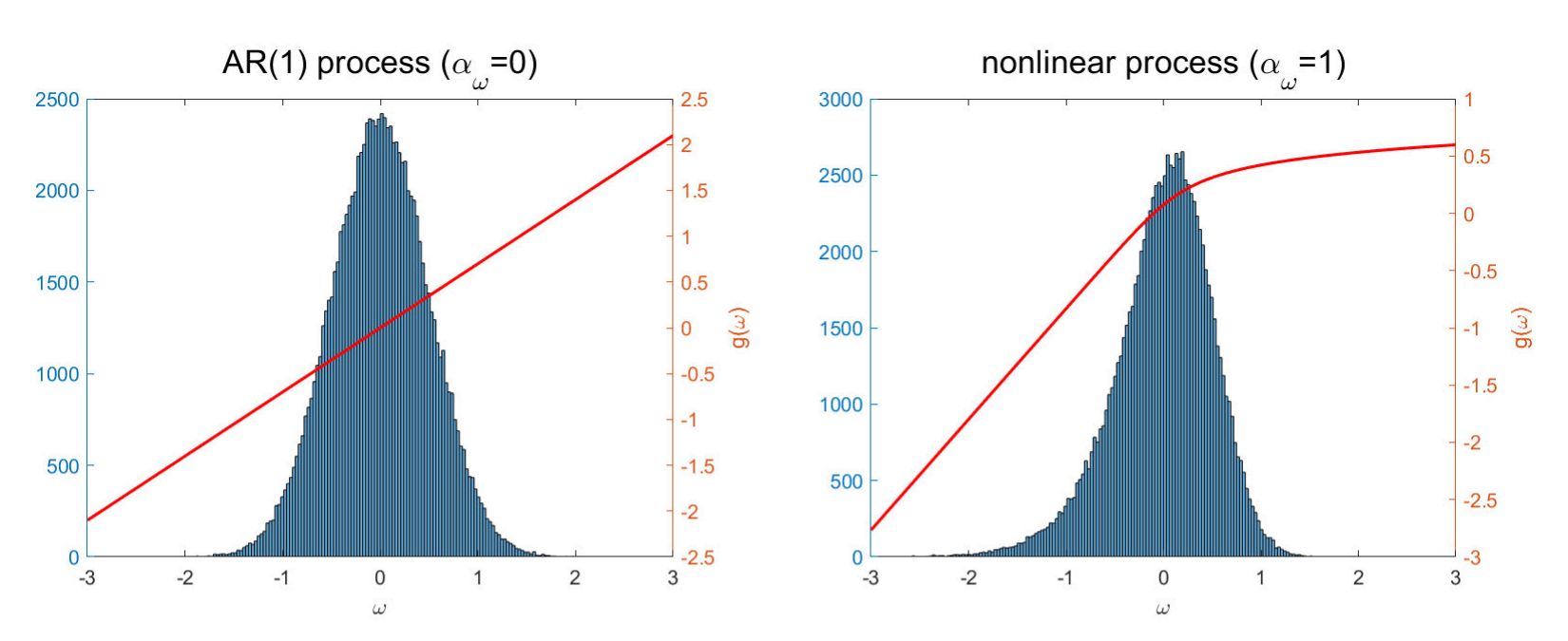}}
\caption{Distribution of productivity (left axis) and law of motion (right axis). Baseline parameterization with Gaussian $AR(1)$ process ($\alpha_\omega=0$, left panel) and nonlinear process ($\alpha_\omega=1$, right panel).\label{fig:productivity}}
\end{figure}

We specify the CES demand
\begin{equation*}
q^*_{it}=\delta_{1i}-(1+\exp(-\delta_{2i}))p_{it},
\end{equation*}
where $p_{it}$ is the output price and $\delta_{i}=(\delta_{1i},\delta_{2i})$ captures shocks to the demand the firm faces and unobserved rivals. This avoids having to specify the imperfectly competitive environment that the firm operates in. Abstracting from time-series variation for simplicity, we specify $\delta_{1i}\sim N\left(\mu_{\delta_1},\sigma^2_{\delta_1}\right)$ and $\delta_{2i}\sim N\left(\mu_{\delta_2},\sigma^2_{\delta_2}\right)$. 

We denote the price of capital as $p^K_{i}$ and the price of the variable input as $p^V_{i}$. Abstracting from time-series variation, we specify $p^K_i\sim N\left(\mu_{p^K},\sigma^2_{p^K}\right)$ and $p^V_i\sim N\left(\mu_{p^V},\sigma^2_{p^V}\right)$. Assuming short-run profit maximization and equating marginal revenue with marginal cost determines $v_{it}$ along with $q^*_{it}$, $q_{it}$, and $p_{it}$. Note that because $v_{it}$ is a function of $k_{it}$, $p^V_{i}$, $\omega_{it}$, and $\delta_{i}$, the firm's demand for the variable input cannot be inverted to express $\omega_{it}$ as a function of observables: invertibility fails as long as $\sigma^2_{\delta_{1}}>0$ or $\sigma^2_{\delta_{2}}>0$.

Recall that capital $k_{it+1}$ is chosen in period $t$. We assume that capital fully depreciates between periods and endow the firm with static expectations regarding $\omega_{it+1}$. Hence, the firm chooses $k_{it+1}$ in period $t$ given $\omega_{it}$, $\delta_{i}$, $p^K_{i}$, and $p^V_{i}$ {\em as if} $k_{it+1}$ and $v_{it+1}$ are variable inputs.\footnote{Assuming constant returns to scale in production and quadratic adjustment costs to capital, the firm's dynamic programming problem that determines investment can be solved in closed form if the firm is a price-taker in the output market \cite{SYVE:01,VANB:07,ACKE:15}. As this is no longer possible if the firm has market power, we opt for a different, tractable specification of the evolution of capital.} 

\begin{table}
\begin{center}
\begin{tabular}{l|ccccccc}
parameter & $\alpha$ &
$\rho$ &
$\nu$ & 
$\expt\left[\omega_{it}\right]$ & 
$\var\left(\omega_{it}\right)$ &
$\corr\left(\omega_{it},\omega_{it-1}\right)$ & \ldots \\
\hline
value & 0.3 & -1 & 0.95 & 0 & $0.5^2$ & 0.7 & \ldots 
\end{tabular}

\bigskip

\begin{tabular}{ccccccccc}
\ldots & $\mu_{\delta_1}$ & 
$\sigma^2_{\delta_1}$ & 
$\mu_{\delta_2}$ & 
$\sigma^2_{\delta_2}$ & 
$\mu_{p^K}$ & 
$\sigma^2_{p^K}$ & 
$\mu_{p^V}$ & 
$\sigma^2_{p^V}$ \\
\hline
\ldots & 10 & $5^2$ & -1.3543 & $0.5^2$ & 0 & $0.5^2$ & 0 & $0.5^2$
\end{tabular}
\end{center}
\caption{Baseline parameterization.\label{tab:baseline}}
\end{table}

Table \ref{tab:baseline} shows our baseline parameterization. The elasticity of substitution is within the range of estimates in the literature \cite{CHIR:08}, as are the 90:10 percentile ratio and persistence of productivity \cite{DELO:21}. Short-run profit maximization implies the markup $\mu_{it}=\frac{P_{it}}{MC_{it}}=1+\exp(\delta_{2i})$, with $\expt\left[\ln\mu_{it}\right]=0.25$ and $\var\left(\ln\mu_{it}\right)=0.0126$. We simulate $S=1000$ datasets with $N=5,000$ firms and $T=20$ periods.\footnote{We use $T_0=5000$ burn-in periods.} We provide further details on the specification in Appendix \ref{app:montecarlo}.

\paragraph{Estimation.}

We estimate $\expt\left[q_{it}|x_{it}\right]$ in step 1 of the OP/LP/ACF procedure by OLS using the complete set of Hermite polynomials of total degree $4$ in the variables in $x_{it}$ as regressors. To illustrate Theorems \ref{thm:main}--\ref{thm:bound}, we vary the specification of $x_{it}$ as detailed below. As an alternative to OLS, we use a neural network estimator with two hidden layers of 128 neurons. We provide further details on the neural network estimator in Appendix \ref{app:montecarlo}.

We estimate the parameters $\theta=\left(\alpha,\rho,\nu,\mu_\omega,\rho_\omega,\alpha_\omega\right)$ in step 2 using the complete set of Hermite polynomials of total degree $4$ in the variables in $z_{it}=\left(k_{it},k_{it-1},v_{it-1},p^V_{i}\right)$ as instruments.\footnote{We treat the price of the variable input $p^V_{i}$ as observable to focus on demand shocks as the reason for the failure of invertibility. The nonidentification result in \citeasnoun{GAND:20} assumes invertibility and therefore does not apply in our setting.} To illustrate Theorem \ref{thm:neyman_orth} and the advantages of Neyman orthogonality in finite samples, we alternatively base the GMM estimator on moment conditions \eqref{eq:step2} and \eqref{eq:modified_step2}. We provide further details on the GMM estimator in Appendix \ref{app:montecarlo}.

We use three statistics to summarize the results. First, we conduct a Lagrange multiplier test for the true parameter value $\theta^0$. Our test corrects for the plug-in nature of the OP/LP/ACF procedure and clustering in the data. We provide further details on the Lagrange multiplier test in Appendix \ref{app:montecarlo}.

Second, following the production function approach to markup estimation, we use the estimate of $\theta_f$ to estimate the markup $\mu_{it}=\frac{P_{it}}{MC_{it}}$ of firm $i$ in period $t$ as
\begin{equation*}
\ln\mu_{it}+\varepsilon_{it}=p_{it}+q_{it}-p^V_{i}-v_{it}+\ln\frac{\partial f(k_{it},v_{it})}{\partial v_{it}}.
\end{equation*}
The right-hand side is the log of the output elasticity minus the log of the expenditure share of the variable input.\footnote{\citeasnoun{DELO:12} isolate
$\ln\mu_{it}$ on the left-hand side and use the residual from the regression in step 1 of the OP/LP/ACF procedure to estimate $\varepsilon_{it}$ on the right-hand side. However, using equation \eqref{eq:step1new}, the residual is $q_{it}-\expt\left[q_{it}|x_{it}\right]=\varepsilon_{it}+\zeta_{it}$. Having $\varepsilon_{it}$ on the left-hand side thus avoids that $\zeta_{it}$ taints the correlation of the estimated markup with a variable of interest such as the firm's export status or a measure of trade liberalization \cite{DORA:21}.} Noting that the disturbance $\varepsilon_{it}$ averages out as $\expt[\varepsilon_{it}]=0$, we refer to the average of $\ln\mu_{it}+\varepsilon_{it}$ across firms and time simply as the average log markup.

Third, because the law of motion is of primary interest in some applications \cite{AW:08,DELO:13,DORA:07}, we use the estimate of $\theta_g$ to estimate
\begin{equation*}
g^\prime(0)=\rho_\omega\left((1-\alpha_\omega)+\alpha_\omega\frac{1}{2\ln 2}\right)
\end{equation*}
as a measure of persistence in the productivity process.

\paragraph{Results.}

We compare three cases for the observables $x_{it}$ in step 1 of the OP/LP/ACF procedure:
\begin{itemize}
\item {\em Case 1:} $x_{it}=\left(k_{it},v_{it},p^V_{i}\right)$;
\item {\em Case 2:} $x_{it}=\left(k_{it+1},k_{it},v_{it},p^V_{i}\right)$;
\item {\em Case 3:} $x_{it}=\left(k_{it+1},k_{it},v_{it},p_{it},p^V_{i}\right)$.
\end{itemize}
Case 1 is our baseline and corresponds to how the OP/LP/ACF procedure is implemented in the existing literature. In light of Theorems \ref{thm:main} and \ref{thm:bias_reduction}, case 2 ensures $x_{it-1}\supseteq z_{it}$. Case 3 heeds Theorem \ref{thm:bound} and pursues a kitchen sink approach by adding the output price as a covariate that speaks to demand conditions.

\begin{sidewaysfigure}
\centerline{\includegraphics[scale=0.8,page=4]{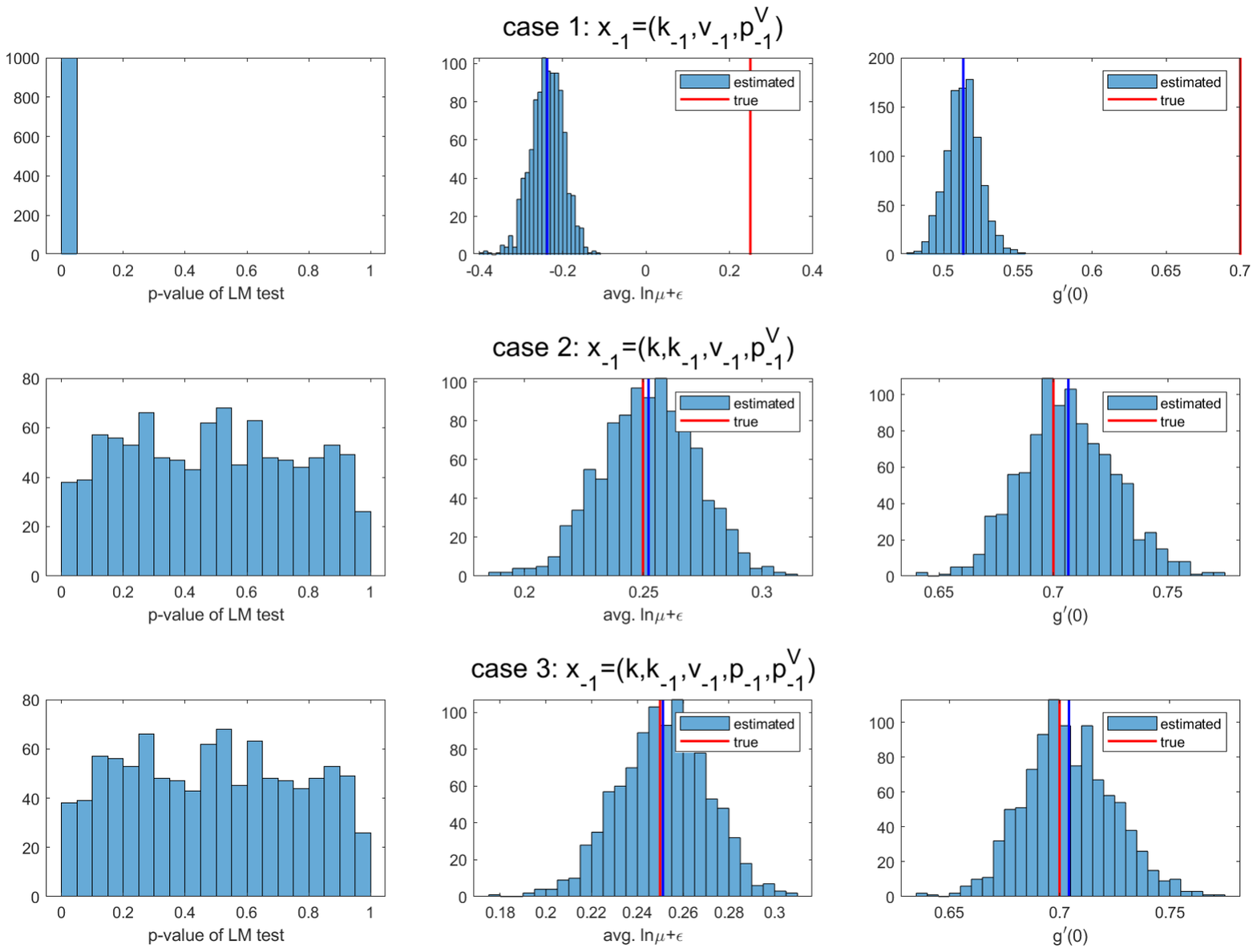}}
\caption{Cases 1, 2, and 3 and moment condition \eqref{eq:step2}. Baseline parameterization with Gaussian $AR(1)$ process ($\alpha_\omega=0$). Case 1 (upper row), case 2 (middle row), and case 3 (lower row). $p$-value for Lagrange multiplier test (left column), average log markup (middle column), and measure of persistence $g^\prime(0)$ (right column).\label{fig:ar1}}
\end{sidewaysfigure}

Figure \ref{fig:ar1} shows the results of the OP/LP/ACF procedure with moment condition \eqref{eq:step2} for the Gaussian $AR(1)$ process ($\alpha_\omega=0$). In case 1, the failure of invertibility causes a massive bias in the production function and in the average log markup derived from it. Indeed, the average log markup centers around $-0.24$ whereas short-run profit maximization implies $\ln\mu_{it}=\ln\left(1+\exp(\delta_{2i})\right)>0$. The failure of invertibility also causes a large bias in the law of motion as summarized by the measure of persistence $g^\prime(0)$. The Lagrange multiplier test for the true parameter value $\theta^0$ rejects by a wide margin. In short, the failure of invertibility can have a substantial, detrimental impact on the OP/LP/ACF procedure as implemented in the existing literature.

\begin{sidewaystable}
\begin{center}
\begin{tabular}{p{10.5cm}|ccc|ccc}
	&	\multicolumn{3}{c|}{average log markup}					&	\multicolumn{3}{c}{measure of persistence $g^\prime(0)$}					\\	
	&	bias	&	variance	&	MSE	&	bias	&	variance	&	MSE	\\	
\hline														
\uline{panel 1: baseline parameterization with $AR(1)$ process, moment condition \eqref{eq:step2}:}	&		&		&		&		&		&		\\	
case 1	&	-0.4875	&	0.0016	&	0.2392	&	-0.1869	&	0.0001	&	0.0351	\\	
case 2	&	0.0023	&	0.0004	&	0.0004	&	0.0065	&	0.0004	&	0.0005	\\	
case 3	&	0.0012	&	0.0004	&	0.0004	&	0.0042	&	0.0004	&	0.0004	\\	
\hline														
\uline{panel 2: baseline parameterization with nonlinear process, moment condition \eqref{eq:step2}:}	&		&		&		&		&		&		\\	
case 1	&	-0.3970	&	0.0015	&	0.1591	&	-0.1524	&	0.0003	&	0.0235	\\	
case 2	&	0.0511	&	0.0005	&	0.0031	&	0.0052	&	0.0003	&	0.0003	\\	
case 3	&	0.0030	&	0.0004	&	0.0004	&	0.0041	&	0.0002	&	0.0002	\\	
\hline														
\uline{panel 3: baseline parameterization with nonlinear process, moment condition \eqref{eq:modified_step2}:}	&		&		&		&		&		&		\\	
case 2	&	0.0415	&	0.0005	&	0.0022	&	-0.0073	&	0.0003	&	0.0003	\\	
case 3	&	0.0037	&	0.0004	&	0.0004	&	0.0028	&	0.0002	&	0.0002	\\	
\hline														
\uline{panel 4: modified parameterization with nonlinear process, moment condition \eqref{eq:step2}:}	&		&		&		&		&		&		\\	
case 2	&	0.5743	&	0.0012	&	0.3309	&	0.0960	&	0.0001	&	0.0093	\\	
case 3	&	-0.0314	&	0.0033	&	0.0043	&	0.0021	&	0.0000	&	0.0000	\\	
\hline														
\uline{panel 5: modified parameterization with nonlinear process, moment condition \eqref{eq:modified_step2}:}	&		&		&		&		&		&		\\	
case 2	&	-0.1327	&	0.0260	&	0.0436	&	-0.0098	&	0.0000	&	0.0001	\\	
case 3	&	-0.0031	&	0.0023	&	0.0023	&	-0.0012	&	0.0000	&	0.0000	\\	
\hline														
\uline{panel 6: baseline parameterization with nonlinear process, neural network estimator, moment condition \eqref{eq:step2}:}	&		&		&		&		&		&		\\	
case 2	&	0.0353	&	0.0202	&	0.0214	&	-0.0209	&	0.0222	&	0.0227	\\	
case 3	&	-0.0196	&	0.0185	&	0.0189	&	-0.0165	&	0.0081	&	0.0084	\\	
\hline														
\uline{panel 7: baseline parameterization with nonlinear process, neural network estimator, moment condition \eqref{eq:modified_step2}:}	&		&		&		&		&		&		\\	
case 2	&	0.0400	&	0.0006	&	0.0022	&	-0.0054	&	0.0003	&	0.0004	\\	
case 3	&	0.0065	&	0.0005	&	0.0005	&	0.0043	&	0.0002	&	0.0002	\\	
\end{tabular}
\end{center}
\caption{Bias, variance, and mean squared error of average log markup and measure of persistence $g^\prime(0)$.\label{tab:summary}}
\end{sidewaystable}

In cases 2 and 3, the Lagrange multiplier test does not reject. As in Example \ref{ex:linear} in Section \ref{sec:validity}, ensuring $x_{it-1}\supseteq z_{it}$ eliminates the bias. The average log markup centers around its true value of $0.25$ and the measure of persistence $g^\prime(0)$ centers around its true value of 0.7. The first panel of Table \ref{tab:summary} summarizes the results. In addition to the bias, it shows the variance and mean squared error of the average log markup and the measure of persistence $g^\prime(0)$.

\begin{sidewaysfigure}
\centerline{\includegraphics[scale=0.8,page=7]{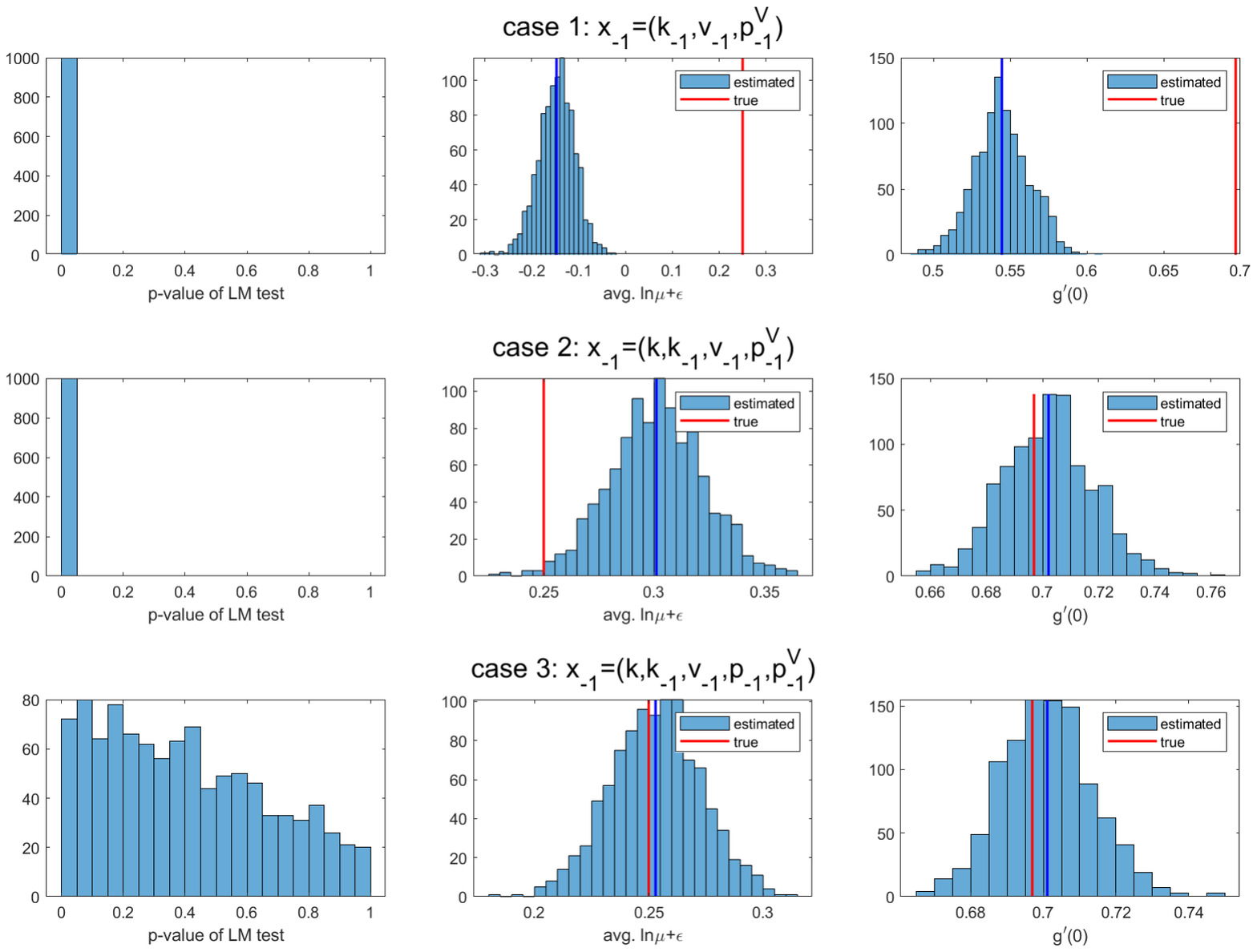}}
\caption{Cases 1, 2, and 3 and moment condition \eqref{eq:step2}. Baseline parameterization with nonlinear process ($\alpha_\omega=1$). Case 1 (upper row), case 2 (middle row), and case 3 (lower row). $p$-value for Lagrange multiplier test (left column), average log markup (middle column), and measure of persistence $g^\prime(0)$ (right column).\label{fig:nonlinear}}
\end{sidewaysfigure}

Turning to the nonlinear process ($\alpha_\omega=1$), Figure \ref{fig:nonlinear} shows the results of the OP/LP/ACF procedure with moment condition \eqref{eq:step2}. In case 1, the failure of invertibility again causes a massive bias in the average log markup and the measure of persistence $g^\prime(0)$. The Lagrange multiplier test for the true parameter value $\theta^0$ rejects by a wide margin. 

Ensuring $x_{it-1}\supseteq z_{it}$ greatly reduces the bias in cases 2 and 3. In case 2, the average log markup centers around $0.30$ compared to its true value of $0.25$. The measure of persistence $g^\prime(0)$ centers around its true value of $0.7$. The Lagrange multiplier test rejects by a wide margin. Adding the output price as a covariate that speaks to demand conditions further reduces---and here almost eliminates---the bias. In case 3, the average log markup centers around $0.25$ and the measure of persistence $g^\prime(0)$ centers around $0.7$. The Lagrange multiplier test rejects for just 72 of the 1000 simulated datasets at a significance level of 5\%. The second panel of Table \ref{tab:summary} summarizes the results.

As discussed in Section \ref{sec:modification_orthogonalization}, explicitly including the first-order bias correction generally mitigates the adverse impact that a noisy estimate in step 1 has on the GMM estimator in step 2, although it can sometimes be redundant if OLS is used in step 1. The third panel of Table \ref{tab:summary} summarizes the results of replacing moment condition \eqref{eq:step2} by moment condition \eqref{eq:modified_step2} for the nonlinear process ($\alpha_\omega=1$). The results are indeed similar for the two moment conditions.\footnote{To confirm, we disrupt the finite-sample relation in equation \eqref{eq:finite_sample_corrective_term} by running OLS on a separate dataset and using the estimates to construct $\hat{e}(x_{it-1})$ for the focal dataset in step 1. For the original moment condition \eqref{eq:step2}, the mean squared error of the average log markup and the measure of persistence $g^\prime(0)$ nearly doubles in case 3. In case 2, the mean squared error of the average log markup decreases slightly while the mean squared error of the measure of persistence $g^\prime(0)$ increases by an order of magnitude. In contrast, there are no meaningful changes for the modified moment condition \eqref{eq:modified_step2}.} 

However, explicitly including the first-order bias correction is not always redundant even if OLS is used in step 1. To illustrate the advantages of Neyman orthogonality in finite samples, we modify the parameterization to $\expt\left[\omega_{it}\right]=-1.25$, $\var\left(\omega_{it}\right)=2^2$, $\corr\left(\omega_{it},\omega_{it-1}\right)=0.85$, $\sigma^2_{\delta_1}=0.5^2$,
$\mu_{\delta_2}=-2.5425$,
$\sigma^2_{\delta_2}=2^2$, and $\sigma^2_{p^K}=2^2$. This implies an increase in the 90:10 percentile ratio and persistence of productivity and a mean-preserving spread of the log markup ($\expt\left[\ln\mu_{it}\right]=0.25$ and $\var\left(\ln\mu_{it}\right)=0.1986$). 

\begin{sidewaysfigure}
\centerline{\includegraphics[scale=0.8,page=10,trim={0 2.5in 0 0},clip]{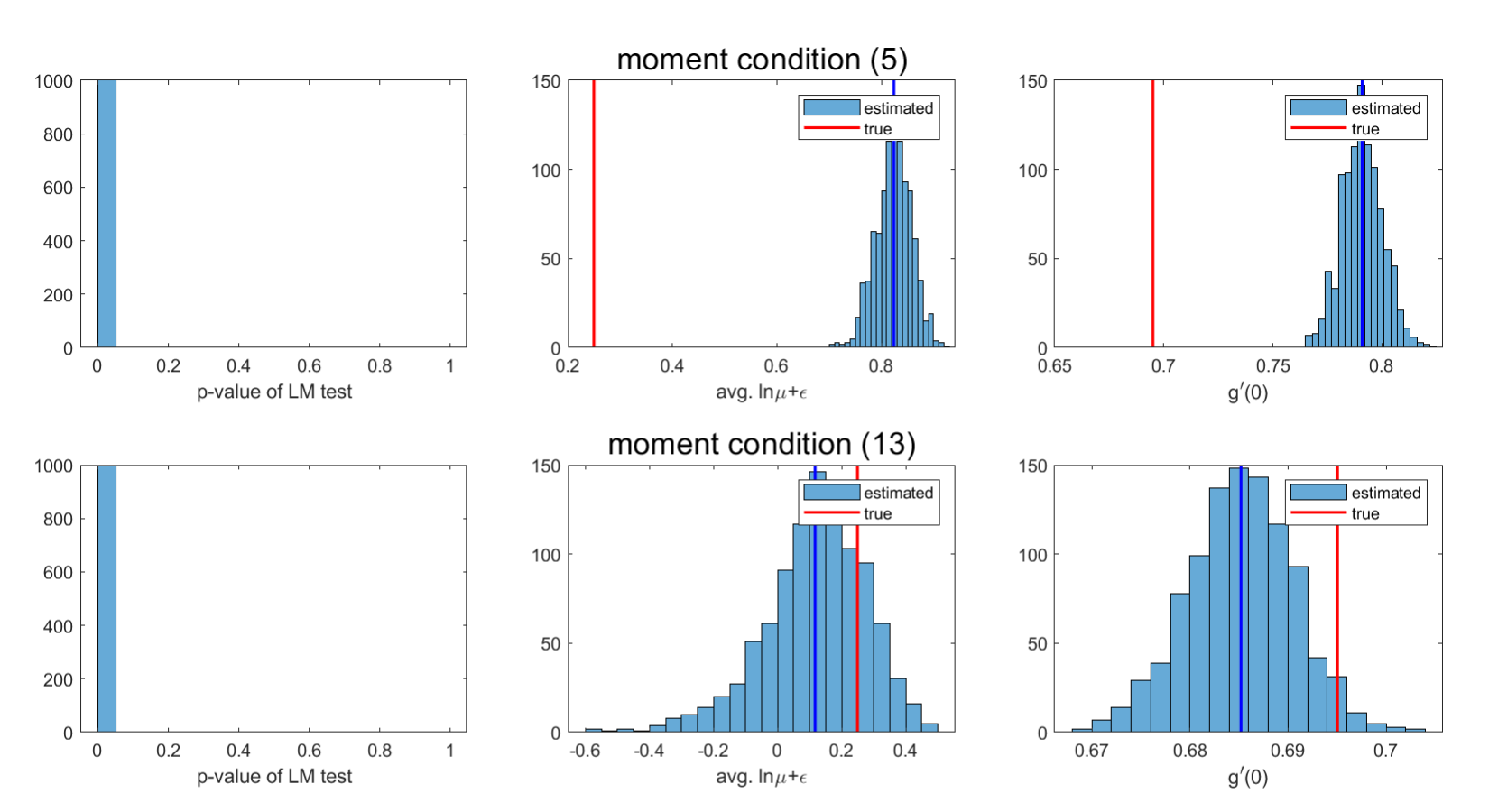}}
\caption{Case 2 and moment conditions \eqref{eq:step2} and \eqref{eq:modified_step2}. Modified parameterization with nonlinear process ($\alpha_\omega=1$). Moment condition \eqref{eq:step2} (upper row) and moment condition \eqref{eq:modified_step2} (lower row). $p$-value for Lagrange multiplier test (left column), average log markup (middle column), and measure of persistence $g^\prime(0)$ (right column).\label{fig:nonlinear_modification_case2}}
\end{sidewaysfigure}

\begin{sidewaysfigure}
\centerline{\includegraphics[scale=0.8,page=11,trim={0 2.5in 0 0},clip]{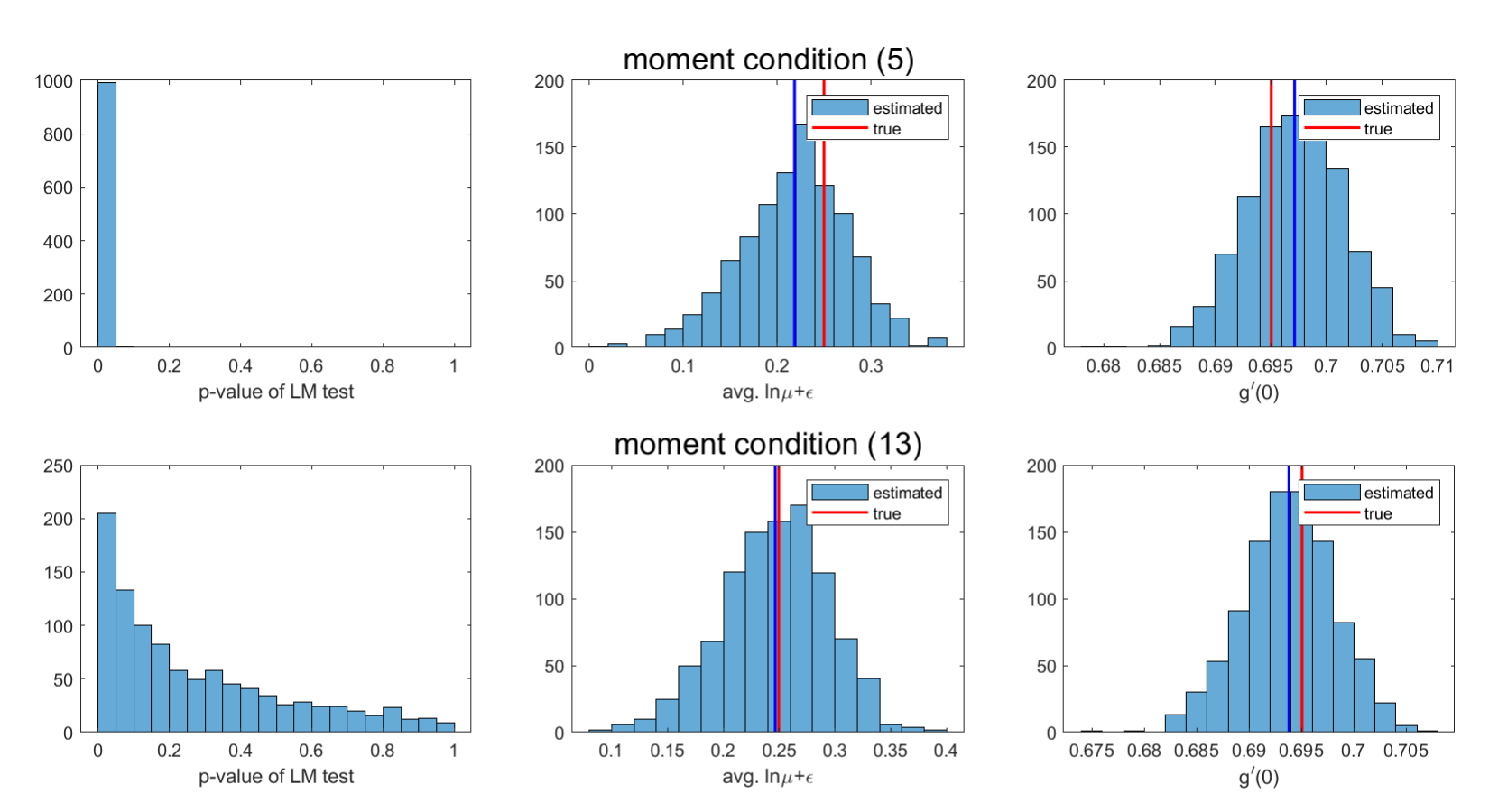}}
\caption{OP/LP/ACF procedure. Case 3 and moment conditions \eqref{eq:step2} and \eqref{eq:modified_step2}. Modified parameterization with nonlinear process ($\alpha_\omega=1$). Moment condition \eqref{eq:step2} (upper row) and moment condition \eqref{eq:modified_step2} (lower row). $p$-value for Lagrange multiplier test (left column), average log markup (middle column), and measure of persistence $g^\prime(0)$ (right column).\label{fig:nonlinear_modification_case3}}
\end{sidewaysfigure}

Figures \ref{fig:nonlinear_modification_case2} and \ref{fig:nonlinear_modification_case3} show the results of the OP/LP/ACF procedure with moment conditions \eqref{eq:step2} and \eqref{eq:modified_step2} for the nonlinear process ($\alpha_\omega=1$). 
Figure \ref{fig:nonlinear_modification_case2} pertains to case 2 and Figure \ref{fig:nonlinear_modification_case3} to case 3.
In case 2, there is a large bias with the original moment condition \eqref{eq:step2}. The average log markup centers around 0.82 compared to its true value of 0.25 and the measure of persistence $g^\prime(0)$ centers around 0.79 compared to its true value of 0.7. The bias is reduced with the modified moment condition \eqref{eq:modified_step2}. The average log markup centers around 0.12 and the measure of persistence $g^\prime(0)$ centers around 0.69. In case 3, the bias is again reduced with the modified moment condition \eqref{eq:modified_step2}. With the original moment condition \eqref{eq:step2}, the average log markup centers around 0.22 and the measure of persistence $g^\prime(0)$ centers around 0.70. The Lagrange multiplier tests rejects by a wide margin. With the modified moment condition \eqref{eq:modified_step2}, the average log markup centers around 0.25 and the measure of persistence $g^\prime(0)$ centers around 0.69. The Lagrange multiplier test rejects for just 205 of the 1000 simulated datasets at a 5\% significance level. The fourth and fifth panels of Table \ref{tab:summary} summarize the results.

Finally, we use a neural network estimator in step 1 as an alternative to OLS. Reverting to the baseline parameterization with the nonlinear process ($\alpha_\omega=1$), the sixth and seventh panels of Table \ref{tab:summary} summarize the results. Comparing the sixth to the second panel of Table \ref{tab:summary} shows that using the neural network estimator increases the variance and thereby the mean squared error of the GMM estimator in step 2 by an order of magnitude if we base the GMM estimator on the original moment condition \eqref{eq:step2}. To the best of our knowledge, there is no theory on using a neural network as a plug-in estimator. In contrast, we achieve much better results if we base the GMM estimator on the modified moment condition \eqref{eq:modified_step2}. Comparing the seventh to the third panel shows that using the neural network estimator is comparable to using OLS. As discussed in Section \ref{sec:modification_orthogonalization}, this aligns with the theory on Neyman orthogonality in the double/debiased machine learning literature as our modification ensures that the GMM estimator in step 2 is oracle efficient regardless of the estimation method used in step 1.

\section{Sensitivity analysis\label{sec:sensitivity}}

While setting $x_{it-1}\supseteq z_{it}$ provides a first-order bias correction, higher-order biases may remain and adversely affect the estimates of $\theta=(\theta_f,\theta_g)$ in step 2 of the OP/LP/ACF procedure. In this section, we develop a diagnostic to assess the sensitivity of the estimates to biases arising from the failure of invertibility. Our diagnostic remains valid without $x_{it-1}\supseteq z_{it}$.

To simplify the analysis, we assume that the nuisance parameter $e(x_{it-1})$ equals its true value $\expt[q_{it-1}|x_{it-1}]$ and abstract from estimation noise. We consider
\begin{equation}\label{eq:step2perturbation}
\expt\left[q_{it} -f(k_{it},v_{it})-g\left(q^*_{it-1} - \lambda\zeta_{it-1} -f(k_{it-1},v_{it-1})\right)|z_{it}\right] = 0,
\end{equation}
where $\lambda\geq 0$ scales the lagged prediction error $\zeta_{it-1}$. If $\lambda = 0$, then moment condition \eqref{eq:step2perturbation} equals the valid but infeasible moment condition \eqref{eq:true_step2}. If $\lambda = 1$, then moment condition \eqref{eq:step2perturbation} equals the original moment condition \eqref{eq:step2} because $\zeta_{it-1} = q^*_{it-1} - \expt[q_{it-1}|x_{it-1}]$ from equation \eqref{eq:step1new}. If, moreover,
$x_{it-1}\supseteq z_{it}$, then it also equals the modified moment condition \eqref{eq:modified_step2} with $e(x_{it-1}) = \expt[q_{it-1}|x_{it-1}]$, as discussed in Section \ref{sec:modification_orthogonalization}.

Define
\begin{equation*}
m_{it}(\theta,\lambda)=q_{it}-f(k_{it},v_{it};\theta_f) -g\left(q^*_{it-1} -\lambda\zeta_{it-1}-f(k_{it-1},v_{it-1};\theta_f);\theta_g\right),
\end{equation*}
where we make explicit the parameterization of the production function and law of motion. Moreover, define the pseudo-true value of $\theta$ as a function of $\lambda$ as
\begin{equation*}\label{eq:GMM_parametric_model}
	\theta(\lambda) \in \argmin_{\theta} \expt[h(z_{it}) m_{it}(\theta, \lambda)]^\top W \expt[h(z_{it}) m_{it}(\theta, \lambda)],
\end{equation*}
where the superscript $\top$ denotes the transpose, $h(z_{it})$ is a vector function of the instruments $z_{it}$, and $W$ is the GMM weighting matrix. By construction, $\theta(0)$ equals the true parameter value $\theta^0$ and $\theta(1)$ is the probability limit of the GMM estimator.

Although it is not possible to consistently estimate the bias $\theta(1)-\theta(0)$, it is possible to consistently estimate $\dd{\theta(\lambda)}{\lambda}\big|_{\lambda = 1}$ and thus to assess the sensitivity at the pseudo-true value $\theta(1)$. A small value of $\dd{\theta(\lambda)}{\lambda}\big|_{\lambda = 1}$ provides assurance that small changes in the prediction error that arises in step 1 of the OP/LP/ACF procedure if invertibility fails do not dramatically alter the pseudo-true value $\theta(1)$. Indeed, one can show that $\dd{\theta(\lambda)}{\lambda}\big|_{\lambda = 1}$ is zero in the examples in Section \ref{sec:validity} where no bias arises despite the failure of invertibility. Conversely, a large value of $\dd{\theta(\lambda)}{\lambda}\big|_{\lambda = 1}$ means that the pseudo-true value $\theta(1)$ is sensitive and thus warns of potentially large bias in the estimates.

To derive our diagnostic $\dd{\theta(\lambda)}{\lambda}\big|_{\lambda = 1}$, we show in Appendix \ref{app:sensitivity} that the pseudo-true value $\theta(\lambda)$ under suitable regularity conditions satisfies the first-order condition
\begin{equation}\label{eq:theta_FOC}
	\expt\left[\pd{m_{it}(\theta(\lambda),\lambda)}{\theta} h^\top(z_{it})\right] W  \expt[h(z_{it}) m_{it}(\theta(\lambda), \lambda)] = 0
\end{equation}
where $\pd{m_{it}(\theta(\lambda),\lambda)}{\theta}$ is a $\dim(\theta)\times 1$ vector. The implicit function theorem implies
\begin{equation}\label{eq:change_theta}
\dd{\theta(\lambda)}{\lambda}\Big|_{\lambda = 1}=-\Gamma^{-1} \gamma ,
\end{equation}
where $\Gamma$ is a $\dim(\theta)\times \dim(\theta)$ matrix and $\gamma$ a $\dim(\theta)\times 1$ vector. Using $\pd{m_{it}}{\theta}$, $\pd{m_{it}}{\lambda}$,  $\pd{^2 m_{it}}{\theta\partial \theta^\top}$, and $\pd{^2 m_{it}}{\theta\partial \lambda}$ to abbreviate $\pd{m_{it}(\theta(\lambda), \lambda)}{\theta}|_{\lambda = 1}$, $\pd{m_{it}(\theta(\lambda), \lambda)}{\lambda}|_{\lambda = 1}$, $\pd{^2 m_{it}(\theta(\lambda), \lambda)}{\theta\partial \theta^\top} |_{\lambda = 1}$, and $\pd{^2 m_{it}(\theta(\lambda), \lambda)}{\theta\partial \lambda} |_{\lambda = 1}$, respectively, the matrix $\Gamma$ is defined as
\begin{align}
	\Gamma = & \left[ \expt\left[ \pd{^2 m_{it}}{\theta\partial \theta^\top} e_1 h^\top(z_{it}) \right] W \expt[h(z_{it})m_{it}(\theta, \lambda)],\ldots,  \expt\left[ \pd{^2 m_{it}}{\theta\partial \theta^\top} e_{\dim(\theta)} h^\top(z_{it}) \right] W \expt[h(z_{it})m_{it}(\theta, \lambda)] \right] \nonumber \\
		 & + \expt\left[ \pd{m_{it}}{\theta} h^\top(z_{it}) \right] W \expt\left[ h(z_{it}) \pd{m_{it}}{\theta^\top}   \right], \label{eq:Gamma_definition}
\end{align}
where $e_l$ is a $\dim(\theta) \times 1$ vector with a one in the $l$th position and zeros elsewhere. The vector $\gamma$ is defined as
\begin{equation}
	\gamma = \expt\left[ \pd{^2 m_{it}}{\theta\partial \lambda} h^\top(z_{it}) \right] W \expt[h(z_{it})m_{it}(\theta, \lambda)] + \expt\left[ \pd{m_{it}}{\theta} h^\top(z_{it}) \right] W \expt\left[ h(z_{it}) \pd{m_{it}}{\lambda}   \right]. \label{eq:gamma_definition}
\end{equation}

To evaluate our diagnostic $\dd{\theta(\lambda)}{\lambda}\big|_{\lambda = 1}$, we assume $\expt[\varepsilon_{it-1}|x_{it-1}, z_{it}] = 0$. As we detail in Appendix \ref{app:sensitivity}, consistently estimating $\Gamma$ is straightforward. Consistently estimating $\gamma$ is complicated by the fact that
\begin{equation*}
\frac{\partial m_{it}}{\partial \lambda} = \frac{\partial g(\expt[q_{it-1}|x_{it-1}] - f(k_{it-1}, v_{it-1}; \theta_f); \theta_g)}{\partial \omega_{it-1}} \zeta_{it-1}
\end{equation*}
and
\begin{equation*}
\pd{^2m_{it}}{\theta\partial \lambda} =
\begin{pmatrix}
	\frac{\partial^2 m_{it}}{\partial \theta_f\partial \lambda} \\[7pt]
\frac{\partial^2 m_{it}}{\partial \theta_g \partial \lambda}
\end{pmatrix}
=
\begin{pmatrix}
	- \frac{\partial^2 g(\expt[q_{it-1}|x_{it-1}] - f(k_{it-1}, v_{it-1}; \theta_f); \theta_g)}{\partial \omega^2_{it-1} } \frac{\partial f(k_{it-1}, v_{it-1}; \theta_f)}{\partial \theta_f} \zeta_{it-1} \\[7pt]
  \frac{\partial^2 g(\expt[q_{it-1}|x_{it-1}] - f(k_{it-1}, v_{it-1}; \theta_f); \theta_g)}{\partial \omega_{it-1}\partial \theta_g} \zeta_{it-1}
\end{pmatrix}
\end{equation*}
depend on the lagged prediction error $\zeta_{it-1}$. The key insight is that assuming $\expt[\varepsilon_{it-1}|x_{it-1}, z_{it}] = 0$ implies
\begin{equation*}
\expt[\zeta_{it-1}|x_{it-1}, z_{it}] = \expt[q_{it-1} - \expt[q_{it-1}|x_{it-1}]\, |x_{it-1}, z_{it}].
\end{equation*}
Using the law of iterated expectations, we can therefore substitute $q_{it-1} - \expt[q_{it-1}|x_{it-1}]$ for $\zeta_{it-1}$ in the above expressions and rewrite $\gamma$ as
\begin{equation}\label{eq:gamma_alternative}
\gamma = \expt\left[ \varphi_{it} h^\top(z_{it}) \right] W \expt[h(z_{it})m_{it}(\theta, \lambda)] + \expt\left[ \pd{m_{it}}{\theta} h^\top(z_{it}) \right] W \expt\left[ h(z_{it}) \psi_{it}   \right],
\end{equation}
where $\frac{\partial m_{it}}{\partial \lambda}$ and $\pd{^2m_{it}}{\theta\partial \lambda}$ in \eqref{eq:gamma_definition} are replaced by
\begin{equation*}
\psi_{it} = \frac{\partial g(\expt[q_{it-1}|x_{it-1}] - f(k_{it-1}, v_{it-1}; \theta_f); \theta_g)}{\partial \omega_{it-1}} (q_{it-1} - \expt[q_{it-1}|x_{it-1}])
\end{equation*}
and
\begin{equation*}
\varphi_{it}
=
\begin{pmatrix}
	 -\frac{\partial^2 g(\expt[q_{it-1}|x_{it-1}] - f(k_{it-1}, v_{it-1}; \theta_f); \theta_g)}{\partial \omega^2_{it-1} } \frac{\partial f(k_{it-1}, v_{it-1}; \theta_f)}{\partial \theta_f} (q_{it-1} - \expt[q_{it-1}|x_{it-1}]) \\[7pt]
 \frac{\partial^2 g(\expt[q_{it-1}|x_{it-1}] - f(k_{it-1}, v_{it-1}; \theta_f); \theta_g)}{\partial \omega_{it-1}\partial \theta_g} (q_{it-1} - \expt[q_{it-1}|x_{it-1}]),
\end{pmatrix}.
\end{equation*}
respectively. We can therefore consistently estimate $\gamma$ by the finite-sample analog to equation \eqref{eq:gamma_alternative} using the estimate of $\expt[q_{it-1}|x_{it-1}]$ from step 1 and that of $\theta$ from step 2.

\begin{table}
\begin{center}
\begin{tabular}{l|ccc|ccc}
	%& \multicolumn{3}{c}{$\theta_f$} & \multicolumn{4}{c}{$\theta_g$}\\
	%\hline
	&	$\alpha$	&	$\rho$	&	$\nu$	&	$\mu_\omega$	&	$\rho_\omega$	&	$\alpha_\omega$	\\
\hline													
\uline{case 1:}	&		&		&		&		&		&		\\
bias	&	0.1306	&	-0.4932	&	0.0090	&	-0.0717	&	-0.2806	&	-0.2659	\\
diagnostic	&	0.1556	&	-0.0739	&	0.0077	&	-0.0013	&	-0.2118	&	-0.1771	\\
	&	(0.0111)	&	(0.1106)	&	(0.0009)	&	(0.0088)	&	(0.0358)	&	(0.1018)	\\
\hline													
\uline{case 2:}	&		&		&		&		&		&		\\
bias	&	-0.0237	&	0.0432	&	-0.0013	&	-0.0789	&	-0.0759	&	-0.2419	\\
diagnostic	&	-0.0034	&	0.0271	&	-0.0012	&	-0.0033	&	-0.0010	&	-0.0541	\\
	&	(0.0011)	&	(0.0103)	&	(0.0002)	&	(0.0010)	&	(0.0028)	&	(0.0110)	\\
\hline													
\uline{case 3:}	&		&		&		&		&		&		\\
bias	&	-0.0034	&	-0.0126	&	-0.0004	&	-0.0054	&	-0.0028	&	-0.0237	\\
diagnostic	&	-0.0003	&	-0.0004	&	-0.0002	&	-0.0003	&	0.0006	&	-0.0045	\\
	&	(0.0009)	&	(0.0071)	&	(0.0001)	&	(0.0014)	&	(0.0025)	&	(0.0082)	\\
\end{tabular}
\end{center}
\caption{Bias and diagnostic $\dd{\theta(\lambda)}{\lambda}\big|_{\lambda = 1}$. Average with standard deviation in brackets. Baseline parameterization with nonlinear process ($\alpha_\omega=1$). Bias and diagnostic use estimate of $\theta$ based on moment condition \eqref{eq:step2}.\label{tbl:diagnostic}}
\end{table}

Table \ref{tbl:diagnostic} shows our diagnostic $\dd{\theta(\lambda)}{\lambda}\big|_{\lambda = 1}$ for the baseline parameterization with the nonlinear process ($\alpha_\omega=1$). We average the diagnostic over the $S=1000$ datasets from Section \ref{sec:montecarlo}. For comparison, we also show the bias $\theta(1)-\theta(0)$. 

In case 1, the diagnostic warns of large biases in the production function parameter $\alpha$ and in the law of motion parameters $\rho_\omega$ and $\alpha_\omega$. While the diagnostic reliably captures the sign of the bias, it is an order of magnitude smaller than the bias in the production function parameter $\rho$. In case 2, the bias and the diagnostic are both smaller than in case 1, with the exception of the law of motion parameter $\alpha_\omega$. The diagnostic accurately reflects the remaining bias in the production function parameter $\rho$. However, it is an order of magnitude smaller than the bias in $\alpha$, $\rho_\omega$, and $\alpha_\omega$. In case 3, the small biases go hand-in-hand with small values of the diagnostic. Overall, while the diagnostic does not consistently estimate the bias $\theta(1)-\theta(0)$, it usefully signals the sensitivity of the estimates to a failure of invertibility.

Finally, we emphasize that our sensitivity analysis is conducted around the potentially misspecified model at $\lambda=1$ as opposed to the true model at $\lambda=0$ and remains valid regardless of the magnitude of misspecification. Our approach thus differs from the sensitivity analysis in \citeasnoun{ANDR:17} that is only valid locally around the true model.

\section{Concluding remarks\label{sec:conclusion}}

The OP/LP/ACF procedure for production function estimation hinges on an invertibility assumption. We show that this assumption is testable and strongly rejected in widely used panel data. A failure of invertibility has important consequences: the prediction of planned output $q^*_{it}$ from observables $x_{it}$ in step 1 of the OP/LP/ACF procedure contains an error that invalidates capital $k_{it}$ as an instrument, leading to biased estimates in step 2.

Fortunately, much can still be done. We establish a necessary and sufficient condition for the moment condition in step 2 of the OP/LP/ACF procedure to hold for the true production function. This condition compels a rethinking of the OP/LP/ACF procedure: any instrument used in step 2 must be appropriately  included in the regression in step 1 to ensure $x_{it-1} \supseteq z_{it}$. At a minimum, this calls for adding the lead of capital $k_{it+1}$ to $x_{it}$. Our condition further suggests to flexibly model the law of motion in step 2 and to take a kitchen sink approach to the regression in step 1 by adding as many relevant covariates as possible.

In case our condition is violated, we show that setting $x_{it-1} \supseteq z_{it}$ provides a first-order bias correction. Explicitly incorporating a bias correction achieves Neyman orthogonality in the modified moment condition. Neyman orthogonality ensures that the asymptotic distribution of the GMM estimator in step 2 is invariant to estimation noise from step 1 and is particularly advantageous if the regression in the step 1 includes a large number of covariates or if modern machine learning techniques such as neural networks and random forests are used. Monte Carlo simulations confirm that Neyman orthogonality can substantially enhance the performance of the GMM estimator in step 2.

Finally, we recognize that, despite ensuring $x_{it-1}\supseteq z_{it}$, higher-order biases may remain. To gauge their importance for the GMM estimator in step 2, we introduce a diagnostic that measures the sensitivity of the estimates to the size of the prediction error that arises in step 1 if invertibility fails.

In sum, the invertibility assumption is demanding and can fail because unobserved demand heterogeneity or in imperfectly competitive environments with partially or fully unobserved rivals or changes in firm conduct and for other reasons. We provide tools for testing the invertibility assumption and propose straightforward modifications of the OP/LP/ACF procedure that either eliminate or mitigate the bias that arises from a failure of invertibility. To address the challenges associated with including a large number of covariates into the regression in step 1, we provide a modified moment condition that achieves Neyman orthogonality and enhances efficiency and robustness. Finally, we hope our diagnostic proves valuable to researchers seeking to assess the potential biases from a failure of invertibility in their estimates.

\newpage
\appendix
\section*{Appendix}

\section{Tests for invertibility\label{app:invertibility}}

\paragraph{Tests.}

\citeasnoun{DELG:01} and \citeasnoun{CAI:24} develop nonparametric inference procedures for a mean-independence restriction (Proposition \ref{prop:test1}), while \citeasnoun{SHAH:20} focus on a conditional-independence restriction (Proposition \ref{prop:test2}). These procedures presume i.i.d. data. Instead of adapting them for panel data, we proceed semiparametrically in what follows.

We test $\expt[q_{it}|x_{it}, x_{it-1}] = \expt[q_{it}|x_{it}]$ in Proposition \ref{prop:test1} using the partially linear regression model
\begin{equation}\label{eq:reg_invertibility}
q_{it} = x_{it-1}^\top\beta  + \psi(x_{it}) + e_{it},
\end{equation}
where the superscript $\top$ denotes the transpose and $\psi(x_{it})$ is a sufficiently flexible function. In practice, we specify $\psi(x_{it})=r^\top(x_{it})\gamma$, where $r(x_{it})$ is the complete set of polynomials of a given total degree in the variables in $x_{it}$ and $\gamma$ is a parameter vector. $\expt[q_{it}|x_{it}, x_{it-1}] = \expt[q_{it}|x_{it}]$ implies $\beta=0$. Rejecting  $H_{0}: \beta = 0$ is therefore evidence against invertibility. To implement this test, we use OLS to estimate the partially linear model \eqref{eq:reg_invertibility} and conduct a standard $F$-test with firm-level clustering.

Because the partially linear model \eqref{eq:reg_invertibility} involves a large number of regressors, the standard $F$-test may perform poorly in finite samples. To address this concern, we use the procedure in \citeasnoun{CHER:15} (henceforth CHS) instead of OLS. This allows us to use LASSO for dimensionality reduction while achieving oracle efficiency for the estimator of $\beta$ via Neyman orthogonalization.

We test $q_{it}\independent x_{it-1}|x_{it}$ in Proposition \ref{prop:test2} using quantile regression to estimate the partially linear model \eqref{eq:reg_invertibility}. Let $Q_\tau(q_{it}|x_{it})$ and $Q_\tau(q_{it}|x_{it}, x_{it-1})$ denote the $\tau$th quantile of $q_{it}$ conditional on $x_{it}$ and $(x_{it}, x_{it-1})$, respectively. $q_{it}\independent x_{it-1}|x_{it}$ implies $\beta=0$ and is equivalent to $Q_{\tau}(q_{it}|x_{it}) = Q_{\tau}(q_{it}|x_{it}, x_{it-1})$ for all $\tau\in (0, 1)$. We restrict attention to the median ($\tau=0.5$) and conduct the $F$-test with firm-level clustering in \citeasnoun{PARE:16}.

\paragraph{Data.}

We implement our tests on two widely used datasets: an unbalanced panel of Spanish manufacturing firms from 1990 to 2006 (Encuesta Sobre Estrategias Empresariales, henceforth ESEE) and the balanced panel of US manufacturing industries from 1958 to 2018 (NBER-CES) from \citeasnoun{BECK:21}.

We follow the data cleaning procedures and variable definitions in \citename{DORA:09} \citeyear{DORA:09,DORA:21} for the ESEE data and in \citeasnoun{JAUM:22} and \citeasnoun{JAUM:24} for the NBER-CES data. The available variables include output $q_{it}$, capital $k_{it}$, labor $l_{it}$, materials $m_{it}$, the price of output $p_{it}$, the wage $w_{it}$, and the price of materials $p^M_{it}$ (in logs). In addition, the ESEE data has a market dynamism variable $mdy_{it}$ that indicates shifts in the demand the firm faces (slump, stability, or expansion). Investment is non-zero in the NBER-CES data, so that the investment variable $inv_{it}$ (in logs) is well-defined.

\paragraph{Results.}

\begin{table}
\begin{center}
\begin{tabular}{c|cc|ccc|cc|cc}
& & & \multicolumn{3}{c|}{total degree 2} & \multicolumn{2}{c|}{total degree 3} & \multicolumn{2}{c}{total degree 4}\\
industry & \#obs & \#firms &  OLS & CHS & Median & OLS & CHS  & OLS & CHS \\
\hline
1  & 2365 & 296 & 0.000 & 0.000 & 0.000 & 0.000 & 0.000 & 0.000 & 0.000 \\
2  & 1270 & 147 & 0.004 & 0.043 & 0.046 & 0.008 & 0.032  & 0.012 & 0.033  \\
3  & 2168 & 278 & 0.000 & 0.000 & 0.002 & 0.000 & 0.000  & 0.000 & 0.000 \\
4  & 1411 & 165 & 0.000 & 0.004 & 0.005 & 0.016 & 0.006  & 0.035 & 0.004 \\
5  & 1505 & 197 & 0.000 & 0.004 & 0.012 & 0.000 & 0.000  & 0.000 & 0.000 \\
6  & 1206 & 148 & 0.000 & 0.000 & 0.018 & 0.001 & 0.000  & 0.001 & 0.000 \\
7  & 2455 & 304 & 0.000 & 0.000 & 0.000 & 0.000 & 0.000  & 0.000 & 0.000 \\
8  & 2368 & 317 & 0.000 & 0.000 & 0.000 & 0.000 & 0.000  & 0.000 & 0.000 \\
9  & 1445 & 191 & 0.001 & 0.000 & 0.001 & 0.017 & 0.000  & 0.009 & 0.000  \\
10 & 1414 & 174 & 0.000 & 0.000 & 0.000 & 0.000 & 0.000  & 0.000 & 0.000  \\
\end{tabular}
\end{center}
\caption{$p$-value for test of $H_0: \beta = 0$ with $x_{it}= (k_{it}, l_{it}, m_{it}, p_{it}, w_{it}, p^M_{it})$. ESEE data.\label{tbl:pval_invertibility_ESEE}}
\end{table}

\begin{table}
\begin{center}
\begin{tabular}{c|cc|ccc|cc|cc}
& & & \multicolumn{3}{c|}{total degree 2} & \multicolumn{2}{c|}{total degree 3} & \multicolumn{2}{c}{total degree 4}\\
industry & \#obs & \#firms &  OLS & CHS & Median & OLS & CHS  & OLS & CHS \\
\hline
1  & 2365 & 296 & 0.000 & 0.000 & 0.000 & 0.000 & 0.000  & 0.000 & 0.000  \\
2  & 1270 & 147 & 0.000 & 0.001 & 0.110 & 0.001 & 0.000  & 0.070 & 0.001  \\
3  & 2168 & 278 & 0.000 & 0.000 & 0.000 & 0.000 & 0.000  & 0.000 & 0.000  \\
4  & 1411 & 165 & 0.001 & 0.000 & 0.107 & 0.062 & 0.000  & 0.208 & 0.000  \\
5  & 1505 & 197 & 0.000 & 0.000 & 0.004 & 0.000 & 0.000  & 0.000 & 0.000 \\
6  & 1206 & 148 & 0.000 & 0.000 & 0.017 & 0.000 & 0.000  & 0.015 & 0.000  \\
7  & 2455 & 304 & 0.000 & 0.000 & 0.000 & 0.000 & 0.000  & 0.000 & 0.000  \\
8  & 2368 & 317 & 0.000 & 0.000 & 0.000 & 0.000 & 0.000  & 0.000 & 0.000  \\
9  & 1445 & 191 & 0.000 & 0.000 & 0.009 & 0.000 & 0.000  & 0.007 & 0.000  \\
10 & 1414 & 174 & 0.000 & 0.000 & 0.000 & 0.000 & 0.000  & 0.000 & 0.000  \\
\end{tabular}
\end{center}
\caption{$p$-value for test of $H_0: \beta = 0$ with $x_{it}= (k_{it}, l_{it}, m_{it}, p_{it}, w_{it}, p^M_{it},mdy_{it})$. ESEE data.\label{tbl:pval_invertibility_ESEE2}}
\end{table}

Tables \ref{tbl:pval_invertibility_ESEE} and \ref{tbl:pval_invertibility_ESEE2} present results for the 10 industries in the ESEE data. We show the $p$-value for the test of $H_0: \beta = 0$ for different estimators (OLS, CHS, and median regression) and for different specifications of $\psi(x_{it})$ (total degree 2, 3, and 4).\footnote{Results for median regression are limited to total degree 2 because the Stata package \texttt{qreg2} fails to converge for total degree 3 and 4.} In Table \ref{tbl:pval_invertibility_ESEE}, we use $x_{it}=(k_{it}, l_{it}, m_{it}, p_{it}, w_{it}, p^M_{it})$. We reject invertibility by a wide margin in all industries and for all combinations of estimators and specifications of $\psi(x_{it})$. In Table \ref{tbl:pval_invertibility_ESEE2}, we add the market dynamism variable $mdy_{it}$ to $x_{it}$. While we continue to reject invertibility in 8 industries, we cannot reject invertibility at a 5\% significance level in industries 2 and 4 for some combinations of estimators and specifications of $\psi(x_{it})$. This reinforces that the ability to control for demand is a precondition for invertibility.

We note that all estimated models in Tables \ref{tbl:pval_invertibility_ESEE} and \ref{tbl:pval_invertibility_ESEE2} have an $R^2$ of at least $0.990$ and some have an $R^2$ of $0.998$. Insufficient flexibility in modeling $\psi(x_{it})$ is thus unlikely to compromise our results.

Turning to the NBER-CES data, we alternatively use $x_{it}=(k_{it}, l_{it}, m_{it}, p_{it}, w_{it}, p^M_{it})$ and $x_{it}=(k_{it}, l_{it}, m_{it}, p_{it}, w_{it}, p^M_{it},inv_{it})$. The latter follows OP to invert the firm's demand for investment. Throughout the $p$-value for the test of $H_0: \beta = 0$ is $0.000$ and the $R^2$ exceeds $0.987$.

\section{Failure of invertibility and validity of OP/LP/ACF moment condition\label{app:validity}}

\paragraph{Omitted proofs.}

\begin{proof}[Proof of Theorem \ref{thm:bias_reduction}]
	Using $\zeta_{it-1}=q^*_{it-1} - \expt[q_{it-1}|x_{it-1}]$, we have
	\begin{eqnarray}
&& \expt\left[ g^\prime\left(q^*_{it-1} - \zeta_{it-1} -f(k_{it-1},v_{it-1}) \right)\zeta_{it-1} \big| z_{it}\right]\nonumber\\
&=& \expt\left[ g^\prime\left(\expt[q_{it-1}|x_{it-1}] -f(k_{it-1},v_{it-1}) \right)\zeta_{it-1} \big| z_{it}\right]\nonumber\\
&=& \expt\left[  \expt [g^\prime\left(\expt[q_{it-1}|x_{it-1}] -f(k_{it-1},v_{it-1}) \right)\zeta_{it-1} | x_{it-1}] \big| z_{it}\right]\nonumber\\
&=& \expt\left[ g^\prime\left(\expt[q_{it-1}|x_{it-1}] -f(k_{it-1},v_{it-1}) \right) \expt[\zeta_{it-1} | x_{it-1}] \big| z_{it}\right]\nonumber\\
& = & 0, \label{eq:bias_correction_theorem_part_2}
	\end{eqnarray}
where the second equality uses $x_{it-1}\supseteq z_{it}$ and the law of iterated expectations, the third equality uses that $g^\prime\left(\expt[q_{it-1}|x_{it-1}] -f(k_{it-1},v_{it-1}) \right)$ is a function of $x_{it-1}$ because $x_{it-1}=\left(k_{it-1},v_{it-1},\ldots\right)$, and the last equality uses $\expt[\zeta_{it-1} | x_{it-1}] =0$. Finally, using equation \eqref{eq:bias_correction_theorem_part_2} and $\zeta_{it-1}=q^*_{it-1} - \expt[q_{it-1}|x_{it-1}]$ establishes equation \eqref{eq:first_order_correction}.
\end{proof}

\begin{proof}[Proof of Theorem \ref{thm:bound}]
Because $g^0\in \mathcal{G}$, we have
\begin{eqnarray*}
	&&\inf_{\tilde g\in \mathcal{G}}\norm{\expt\left[g^0(\omega_{it-1})|z_{it}\right]-\expt\left[\tilde g\left(\expt\left[\omega_{it-1}|x_{it-1}\right]\right)|z_{it}\right]}_{L, 1} \\
	&=&\inf_{\tilde g\in \mathcal{G}}\expt\left[\Big|\expt\left[g^0(\omega_{it-1})|z_{it}\right]-\expt\left[\tilde g\left(\expt\left[\omega_{it-1}|x_{it-1}\right]\right)|z_{it}\right]\Big|\right] \\
	&\le& \expt\left[\Big\vert \expt[g^0(\omega_{it-1})|z_{it}] - \expt[g^0(\expt[\omega_{it-1}|x_{it-1}])|z_{it}] \Big \vert\right].
\end{eqnarray*}
Using a second-order Taylor expansion of $g^0$ around $\expt[\omega_{it-1}|x_{it-1}]$, we have
\begin{eqnarray*}
&&\Big\vert \expt[g^0(\omega_{it-1})|z_{it}] - \expt[g^0(\expt[\omega_{it-1}|x_{it-1}])|z_{it}] \Big \vert \\
&=& \Big\vert\expt \left[ g^{0\prime}(\expt[\omega_{it-1}|x_{it-1}])\zeta_{it-1}  \Big\vert z_{it} \right] + \expt \left[ g^{0\prime\prime}(\tilde{\omega}_{it-1}) \zeta_{it-1}^2 \Big\vert z_{it} \right] \Big\vert \\
&\le &\Big\vert\expt \left[ g^{0\prime}(\expt[\omega_{it-1}|x_{it-1}])\zeta_{it-1}  \Big\vert z_{it} \right] \Big\vert +
\Big\vert \expt \left[ g^{0\prime\prime}(\tilde{\omega}_{it-1}) \zeta_{it-1}^2 \Big\vert z_{it} \right] \Big\vert
\\
&\le &\Big\vert\expt \left[ g^{0\prime}(\expt[\omega_{it-1}|x_{it-1}])\zeta_{it-1}  \Big\vert z_{it} \right] \Big\vert + \tau \expt \left[ \zeta_{it-1}^2 \Big\vert z_{it} \right],
\end{eqnarray*}
where $\tilde{\omega}_{it-1}$ is some value between $\omega_{it-1}$ and $\expt[\omega_{it-1}|x_{it-1}]$, and the last inequality follows from the definition of $\tau$. As in the proof of Theorem \ref{thm:bias_reduction}, using $x_{it-1}\supseteq z_{it}$ and the law of iterated expectations yields $\expt \left[ g^{0\prime}(\expt[\omega_{it-1}|x_{it-1}]) \zeta_{it-1} \Big\vert z_{it} \right] = 0 $. Thus, we have
\begin{equation*}
	\inf_{\tilde g\in \mathcal{G}}\norm{\expt\left[g^0(\omega_{it-1})|z_{it}\right]-\expt\left[\tilde g\left(\expt\left[\omega_{it-1}|x_{it-1}\right]\right)|z_{it}\right]}_{L, 1} \le\expt\left[ \tau \expt \left[ \zeta_{it-1}^2 | z_{it} \right]\right] = \tau \var(\zeta_{it-1}).
\end{equation*}
\end{proof}

\section{Monte Carlo exercise\label{app:montecarlo}}

\paragraph{Data generating process.}

Marginal revenue is
\begin{equation*}
mr_{it}=-\ln\left(1+\exp(\delta_{2i})\right)+p_{it}
\end{equation*}
and marginal cost is
\begin{equation*}
mc_{it}=p^V_{i}+v_{it}-f(k_{it},v_{it})-\ln\frac{\partial f(k_{it},v_{it})}{\partial v_{it}}-\omega_{it}.
\end{equation*}
We rewrite $mr_{it}=mc_{it}$ as
$$
\frac{\delta_{1i}+\exp(-\delta_{2i})(f(k_{it},v_{it})+\omega_{it})}{1+\exp(-\delta_{2i})}-\ln\left(1+\exp(\delta_{2i})\right)+\ln\frac{\partial f(k_{it},v_{it})}{\partial v_{it}}-p^V_{i}-v_{it}=0
$$
and solve numerically for $v_{it}$. With $v_{it}$ in hand, we determine $q^*_{it}=f(k_{it},v_{it})+\omega_{it}$, $q_{it}=q^*_{it}+\varepsilon_{it}$, and $p_{it}=\frac{\delta_{1i}-q^*_{it}}{1+\exp(-\delta_{2i})}$.

As capital fully depreciates between periods and the firm has static expectations, it chooses $k_{it+1}$ in period $t$ given $\omega_{it}$, $\delta_{i}$, $p^K_{i}$, and $p^V_{i}$ whilst anticipating that it optimally chooses $\tilde v_{it+1}$ in period $t+1$. To maximize its expected profit in period $t+1$, in period $t$ the firm solves
\begin{gather*}
\max_{k_{it+1}}\left(\max_{\tilde v_{it+1}}\exp(\tilde p_{it+1}+\tilde q^*_{it+1})-\exp(p^K_{i}+k_{it+1})-\exp(p^V_{i}+\tilde v_{it+1})\right) \\
=\max_{k_{it+1},\tilde v_{it+1}}\exp(\tilde p_{it+1}+\tilde q^*_{it+1})-\exp(p^K_{i}+k_{it+1})-\exp(p^V_{i}+\tilde v_{it+1})
\end{gather*}
subject to
\begin{gather*}
\tilde q^*_{it+1}=f(k_{it+1},\tilde v_{it+1})+\omega_{it}, \\
\tilde q^*_{it+1}=\delta_{1i}-(1+\exp(-\delta_{2i}))\tilde p_{it+1}.
\end{gather*}
The solution is
\begin{gather*}
k_{it+1}=\frac{1}{1-\rho}\left(\ln\alpha-p^K_{i}\right)
	+\frac{\exp(\delta_{2i})+1}{\exp(\delta_{2i})+1-\nu}\Bigg(
	\ln\nu \notag\\
	+\left(\frac{\nu}{(\exp(\delta_{2i})+1)\rho}-1\right)\ln\left(\alpha\left(\frac{\alpha}{\exp(p^K_{i})}\right)^\frac{\rho}{1-\rho}
	+(1-\alpha)\left(\frac{1-\alpha}{\exp(p^V_{i})}\right)^\frac{\rho}{1-\rho}\right) \notag\\
	-\ln\left(\exp(\delta_{2i})+1\right)+\frac{1}{\exp(\delta_{2i})+1}\omega_{it}+\frac{\delta_{1i}}{1+\exp(-\delta_{2i})}
	\Bigg).
\end{gather*}

\paragraph{Neural network estimator.}

The neural network uses two hidden layers of 128 neurons and ReLU activation. We standardize the input data. We train the network using the stochastic gradient decent (SDG) method on 80\% of firms, setting aside 20\% of firms for validation, to minimize the mean squared error. We stop the SDG iterations if the mean squared error for the validation sample does not improve after 10 epoch iterations. We use the trainnet command in Matlab and set the learning rate to 0.01 and the mini-batch size to 500 observations.

\paragraph{GMM estimator.}

Corresponding to moment condition \eqref{eq:step2} in step 2, define the moment function
\begin{equation}\label{eq:moment_791}
m_{it}(\theta)=q_{it}-f(k_{it},v_{it};\theta_f)-g\left(e(x_{it-1}) - f(k_{it-1},v_{it-1};\theta_f);\theta_g\right),
\end{equation}
where we replace $e(x_{it-1})$ by the estimate from step 1. In step 2, we solve the GMM problem
\begin{equation}\label{eq:GMM_795}
\min_\theta \left(\frac{1}{NT}\sum_{i,t}h(z_{it})m_{it}(\theta)\right)^\top W \left(\frac{1}{NT}\sum_{i,t}h(z_{it})m_{it}(\theta)\right),
\end{equation}
where the superscript $\top$ denotes the transpose, $h(z_{it})$ is the complete set of Hermite polynomials of total degree 4 in the variables in $z_{it}$, and 
\begin{equation*}
	W=\left(\frac{1}{NT-1}\sum_{i,t}\left(h(z_{it})m_{it}\left(\theta^0\right) - \hat{\mu}\right)^\top\left(h(z_{it})m_{it}\left(\theta^0\right) - \hat{\mu}\right)\right)^{-1} \text{ with } \hat{\mu} = \frac{1}{NT}\sum_{i,t}h(z_{it})m_{it}\left(\theta^0\right),
\end{equation*}
is a weighting matrix evaluated at the true parameters values $\theta^0$.

Corresponding to moment condition \eqref{eq:modified_step2}, our modification of the OP/LP/ACF procedure redefines the moment function as
\begin{multline}\label{eq:modified_moment_791}
m_{it}(\theta)=
q_{it}-f(k_{it},v_{it};\theta_f) - g\left(e(x_{it-1}) -f(k_{it-1},v_{it-1};\theta_f);\theta_g \right) \\
	 - g^\prime\left(e(x_{it-1}) -f(k_{it-1},v_{it-1};\theta_f);\theta_g \right) (q_{it-1} - e(x_{it-1})),
\end{multline}
where we replace $e(x_{it-1})$ by the estimate from step 1. We hold fixed $W$ to facilitate the comparison between our modification and the original procedure. 

\paragraph{Lagrange multiplier (LM) test.}

We first derive the LM test for the original moment function \eqref{eq:moment_791}. To account for the plug-in nature of the GMM estimator in step 2, we have to explicitly specify how $e(x_{it-1})$ is modelled and estimated in step 1. In our Monte Carlo exercise, we model $e(x_{it-1}) = r^\top(x_{it-1}) \tau$, where the superscript $\top$ denotes the transpose, $r(x_{it-1})$ is the complete set of Hermite polynomials of total degree 4 in the variables in $x_{it-1}$, and $\tau$ is a parameter vector. Consequently, we define
\begin{equation}\label{eq:moment_816}
m_{it}(\theta, \tau)=q_{it}-f(k_{it},v_{it};\theta_f)-g\left(r^\top(x_{it-1}) \tau - f(k_{it-1},v_{it-1};\theta_f);\theta_g\right),
\end{equation}
where we make explicit the dependence of the moment function on $\tau$. We estimate $\tau$ in step 1 by OLS as $\hat{\tau} = [\frac{1}{NT} \sum_{i,t} r(x_{it-1}) r^\top (x_{it-1})]^{-1} \frac{1}{NT}\sum_{i,t} q_{it-1} r(x_{it-1})$.\footnote{We can alternatively estimate $\tau$ as $\hat{\tau} = [\frac{1}{NT} \sum_{i,t} r(x_{it}) r^\top (x_{it})]^{-1} \frac{1}{NT}\sum_{i,t} q_{it} r(x_{it})$. In our simulated datasets, both current and lagged values are available for all $T$ periods and $N$ firms.}

Following standard asymptotic arguments, we can show that
\begin{equation*}
\widehat{\Sigma}^{-\frac{1}{2}}\frac{1}{\sqrt{NT}} \sum_{i, t} h(z_{it}) m_{it}(\theta^0, \hat{\tau}) \convd N(0, I),
\end{equation*}
where $\widehat{\Sigma} = \widehat{\Lambda} \widehat{\Omega} \widehat{\Lambda}^{\top}$,
\begin{equation*}
	\widehat{\Lambda} = \begin{pmatrix}
		I & \sum_{i,t} h(z_{it})\left[ \frac{\partial m_{it}(\theta^0, \hat{\tau})}{\partial \tau^\top} \right] \left[ \sum_{i,t} r(x_{it-1}) r^\top(x_{it-1})  \right]^{-1}
\end{pmatrix},
\end{equation*}
with $I$ being the $\dim(h(z_{it}))\times \dim(h(z_{it}))$ identity matrix, and
\begin{eqnarray*}
	\widehat{\Omega} &=&
\begin{pmatrix}
	\widehat{\Omega}_{11} & \widehat{\Omega}_{12}\\
	\widehat{\Omega}^\top_{12} & \widehat{\Omega}_{22}
\end{pmatrix}, \\
	\widehat{\Omega}_{11} &=& \frac{1}{NT}\sum_{t=1}^T\sum_{t'=1}^T\sum_{i=1}^N\left[ m_{it}(\theta^0, \hat{\tau}) m_{it'}(\theta^0, \hat{\tau}) h(z_{it}) h^\top(z_{it'}) \right], \\
	\widehat{\Omega}_{12} &=& \frac{1}{NT}\sum_{t=1}^T\sum_{t'=1}^T\sum_{i=1}^N\left[ m_{it}(\theta^0, \hat{\tau}) (q_{it'-1} - r^\top(x_{it'-1}) \hat{\tau}) h(z_{it}) r^\top(x_{it'-1}) \right], \\
	\widehat{\Omega}_{22} &=& \frac{1}{NT}\sum_{t=1}^T\sum_{t'=1}^T\sum_{i=1}^N\left[ (q_{it-1} - r^\top(x_{it-1}) \hat{\tau}) (q_{it'-1} - r^\top(x_{it'-1}) \hat{\tau}) r(x_{it-1}) r^\top(x_{it'-1}) \right].
\end{eqnarray*}
Note that $\widehat{\Sigma}$ incorporates firm-level clustering and thus allows for arbitrary correlation within firm across time in the data.
As $N\to\infty$, the LM statistic
\begin{equation*}
S_{N} = \frac{1}{NT} \left( \sum_{i, t} h(z_{it}) m_{it}(\theta^0, \hat{\tau}) \right)^\top \widehat{\Sigma}^{-1} \left( \sum_{i, t} h(z_{it}) m_{it}(\theta^0, \hat{\tau}) \right)
\end{equation*}
converges to a $\chi^2$ distribution with $\dim(h(z_{it}))$ degrees of freedom. Let $c_{1-\alpha}$ be the $1-\alpha$ quantile of this distribution. Then rejecting $H_0: \expt[h(z_{it})m_{it}(\theta^0, \tau^0)] = 0$ if $S_N > c_{1-\alpha}$ is a valid LM inference procedure with asymptotic size $\alpha$.

Turning to the modified moment function \eqref{eq:modified_moment_791}, Neyman orthogonality renders correcting for the plug-in nature of the GMM estimator in step 2 superfluous. Consequently, we use a standard LM test with firm-level clustering. As $N\to\infty$, the LM statistic
\begin{equation*}
S_{N} = \frac{1}{NT} \left( \sum_{i, t} h(z_{it}) m_{it}(\theta^0, \hat{\tau}) \right)^\top \widehat{\Omega}_{11}^{-1} \left( \sum_{i, t} h(z_{it}) m_{it}(\theta^0, \hat{\tau}) \right)
\end{equation*}
converges to a $\chi^2$ distribution with $\dim(h(z_{it}))$ degrees of freedom. Rejecting $H_0: \expt[h(z_{it})m_{it}(\theta^0, \tau^0)] = 0$ if $S_N > c_{1-\alpha}$ is a valid LM inference procedure with asymptotic size $\alpha$.

\section{Modification of OP/LP/ACF moment condition\label{app:modification_orthogonalization}}

\paragraph{Omitted proofs.}

\begin{proof}[Proof of Theorem \ref{thm:neyman_orth}]
	By the definition of the derivative, it suffices to show that for any non-zero sequence $\lambda_n$ with $\lim_{n\to\infty}\lambda_n = 0$, we have
	\begin{multline}
		\lim_{n\to\infty} \frac{1}{\lambda_n} \Bigg[ \Bigg(\expt\left[q_{it}-f(k_{it},v_{it}) - g\left( \expt[q_{it-1}|x_{it-1}] + \lambda_n \delta(x_{it-1}) -f(k_{it-1},v_{it-1}) \right) \right.\\
		- \left. g^\prime\left(\expt[q_{it-1}|x_{it-1}] + \lambda_n \delta(x_{it-1}) -f(k_{it-1},v_{it-1}) \right)(q_{it-1} - \expt[q_{it-1}|x_{it-1}] - \lambda_n \delta(x_{it-1}) \big| z_{it}\right] \Bigg)\\
- \Bigg(\expt\big[q_{it}-f(k_{it},v_{it}) - g\left( \expt[q_{it-1}|x_{it-1}]  -f(k_{it-1},v_{it-1}) \right) \\
-  g^\prime\left(\expt[q_{it-1}|x_{it-1}]  -f(k_{it-1},v_{it-1}) \right)(q_{it-1} - \expt[q_{it-1}|x_{it-1}]) \big| z_{it}\big] \Bigg) \Bigg] = 0. \label{eq:722}
	\end{multline}

Similar to the idea used in the proof of Theorem \ref{thm:bias_reduction}, the law of iterated expectations implies
    \begin{gather}
	\expt\left[  g^\prime\left(\expt[q_{it-1}|x_{it-1}] + \lambda_n \delta(x_{it-1}) -f(k_{it-1},v_{it-1}) \right)(q_{it-1} - \expt[q_{it-1}|x_{it-1}]) | z_{it} \right] = 0, \label{eq:727} \\
	\expt\left[  g^\prime\left(\expt[q_{it-1}|x_{it-1}] -f(k_{it-1},v_{it-1}) \right)(q_{it-1} - \expt[q_{it-1}|x_{it-1}]) | z_{it} \right] = 0. \label{eq:730}
	\end{gather}
Using equations \eqref{eq:727} and \eqref{eq:730}, equation \eqref{eq:722} is equivalent to
\begin{align}
	  & \lim_{n\to\infty}\frac{1}{\lambda_n}\expt\left[g\left( \expt[q_{it-1}|x_{it-1}] + \lambda_n \delta(x_{it-1}) -f(k_{it-1},v_{it-1}) \right)  - g\left( \expt[q_{it-1}|x_{it-1}]  -f(k_{it-1},v_{it-1}) \right) \Big\vert z_{it} \right] \nonumber\\
	=& \lim_{n\to\infty }\expt \left[ g^\prime\left(\expt[q_{it-1}|x_{it-1}] + \lambda_n \delta(x_{it-1}) -f(k_{it-1},v_{it-1}) \right)\delta(x_{it-1}) \Big\vert z_{it}\right].  \label{eq:749}
\end{align}
It therefore remains to establish equation \eqref{eq:749}. We do this by showing that both sides of equation \eqref{eq:749} are equal to $\expt\left[  g^\prime\left( \expt[q_{it-1}|x_{it-1}]  -f(k_{it-1},v_{it-1}) \right) \delta(x_{it-1}) \bigg\vert z_{it} \right]$.

Because $g^\prime$ is bounded and $\delta$ is integrable,  there exists some $\eta > 0$ such that
\begin{equation}\label{eq:753}
	\expt\left[ \sup_{\lambda: |\lambda| < \eta } \left| g^\prime\left(\expt[q_{it-1}|x_{it-1}] + \lambda \delta(x_{it-1}) -f(k_{it-1},v_{it-1}) \right)\delta(x_{it-1}) \right|   \Bigg\vert z_{it} \right] < \infty
\end{equation}
almost surely.\footnote{Theorem \ref{thm:neyman_orth} continues to hold if the assumption that $g^\prime$ is bounded is relaxed to this condition.}
Because $g^\prime $ is continuous, the dominated convergence theorem therefore implies that
\begin{eqnarray}
&&\lim_{n\to\infty }\expt \left[ g^\prime\left(\expt[q_{it-1}|x_{it-1}] + \lambda_n \delta(x_{it-1}) -f(k_{it-1},v_{it-1}) \right)\delta(x_{it-1}) \Big\vert z_{it}\right] \nonumber \\
&=&\expt \left[ g^\prime\left(\expt[q_{it-1}|x_{it-1}] -f(k_{it-1},v_{it-1}) \right)\delta(x_{it-1}) \Big\vert z_{it}\right]. \label{eq:758}
\end{eqnarray}

Turning to the other side of equation \eqref{eq:749}, the mean value theorem implies that
\begin{align}
&\frac{1}{\lambda_n}\left( g\left( \expt[q_{it-1}|x_{it-1}] + \lambda_n \delta(x_{it-1}) -f(k_{it-1},v_{it-1}) \right)  - g\left( \expt[q_{it-1}|x_{it-1}]  -f(k_{it-1},v_{it-1}) \right)  \right) \nonumber \\
	=& g'\left( \expt[q_{it-1}|x_{it-1}] + \tilde{\lambda}_n \delta(x_{it-1}) -f(k_{it-1},v_{it-1}) \right) \delta(x_{it-1}),
\end{align}
where $\tilde{\lambda}_n$ is some value between $0$ and $\lambda_n$ that implicitly depends on $x_{it-1}$. Because $|\tilde{\lambda}_n| \le |\lambda_n|$ and $g^\prime$ is continuous, we have
\begin{align}
&\lim_{n\to\infty} g'\left( \expt[q_{it-1}|x_{it-1}] + \tilde{\lambda}_n \delta(x_{it-1}) -f(k_{it-1},v_{it-1}) \right) \delta(x_{it-1}) \nonumber \\
	= &\  g'\left( \expt[q_{it-1}|x_{it-1}]  -f(k_{it-1},v_{it-1}) \right) \delta(x_{it-1}) \label{eq:768}
\end{align}
for all $x_{it-1}$. Given inequality \eqref{eq:753} and equation \eqref{eq:768}, the dominated convergence theorem implies that
\begin{align*}
&\lim_{n\to\infty}\frac{1}{\lambda_n}\expt\left[g\left( \expt[q_{it-1}|x_{it-1}] + \lambda_n \delta(x_{it-1}) -f(k_{it-1},v_{it-1}) \right)  - g\left( \expt[q_{it-1}|x_{it-1}]  -f(k_{it-1},v_{it-1}) \right) \Big\vert z_{it} \right] \\
	=&\  \expt\left[  g^\prime\left( \expt[q_{it-1}|x_{it-1}]  -f(k_{it-1},v_{it-1}) \right) \delta(x_{it-1}) \bigg\vert z_{it} \right].
\end{align*}
This establishes equation \eqref{eq:749}.
\end{proof}

\section{Sensitivity analysis\label{app:sensitivity}}

Assume $f(k_{it}, v_{it}; \theta_f)$ is twice continuously differentiable in $\theta_f$ and $g(\omega_{it-1}; \theta_g)$ is twice continuously differentiable in $(\theta_g, \omega_{it-1})$, so that  $m_{it}(\theta, \lambda)$ is twice continuously differentiable in $(\theta, \lambda)$. Assume the product of $h$ and the first-order derivatives of $m_{it}(\theta, \lambda)$ is always integrable. Then the dominated convergence theorem implies
\begin{equation*}
	\pd{\expt[h(z_{it}) m_{it}(\theta, \lambda)]}{\theta^\top} = \expt\left[h(z_{it}) \pd{m_{it}(\theta, \lambda)}{\theta^\top} \right],\quad \pd{\expt[h(z_{it}) m_{it}(\theta, \lambda)]}{\lambda} = \expt\left[h(z_{it}) \pd{m_{it}(\theta, \lambda)}{\lambda} \right].
\end{equation*}
Similarly, assume the product of $h$ and the second-order derivatives of $m_{it}(\theta, \lambda)$ is always integrable. Then the dominated convergence theorem implies
\begin{equation*}
	\pd{^2\expt[h(z_{it}) m_{it}(\theta, \lambda)]}{\theta \partial\theta^\top} = \expt\left[h(z_{it}) \pd{^2m_{it}(\theta, \lambda)}{\theta \partial\theta^\top} \right],\quad	\pd{^2\expt[h(z_{it}) m_{it}(\theta, \lambda)]}{\theta \partial\lambda} = \expt\left[h(z_{it}) \pd{^2m_{it}(\theta, \lambda)}{\theta \partial\lambda} \right].
\end{equation*}
Under these regularity conditions, the pseudo-true value $\theta(\lambda)$ satisfies the first-order condition \eqref{eq:theta_FOC}.

Assume the first-order condition \eqref{eq:theta_FOC} identifies $\theta(\lambda)$ for $\lambda$ in a neighborhood of $1$. Then the implicit function theorem implies
\begin{multline}\label{eq:828}
	\expt\left[ \left( \pd{^2 m_{it}}{\theta\partial \theta^\top} \dd{\theta(\lambda)}{\lambda}\Big|_{\lambda = 1} + \pd{^2 m_{it}}{\theta\partial \lambda} \right)h^\top(z_{it})  \right] W \expt\left[ h(z_{it})m_{it}(\theta(\lambda), \lambda)  \right]\\
	+ \expt\left[ \pd{m_{it}}{\theta} h^\top(z_{it}) \right] W \expt\left[ h(z_{it})  \left( \pd{m_{it}}{\theta^\top} \dd{\theta(\lambda)}{\lambda}\Big|_{\lambda = 1} + \pd{m_{it}}{\lambda} \right) \right] = 0,
\end{multline}
where we use $\pd{m_{it}}{\theta}$, $\pd{m_{it}}{\lambda}$,  $\pd{^2 m_{it}}{\theta\partial \theta^\top}$, and $\pd{^2 m_{it}}{\theta\partial \lambda}$ to abbreviate $\pd{m_{it}(\theta(\lambda), \lambda)}{\theta}|_{\lambda = 1}$, $\pd{m_{it}(\theta(\lambda), \lambda)}{\lambda}|_{\lambda = 1}$, $\pd{^2 m_{it}(\theta(\lambda), \lambda)}{\theta\partial \theta^\top} |_{\lambda = 1}$, and $\pd{^2 m_{it}(\theta(\lambda), \lambda)}{\theta\partial \lambda} |_{\lambda = 1}$, respectively.
Equation \eqref{eq:828} can be rewritten as
\begin{equation*}
\Gamma  \dd{\theta(\lambda)}{\lambda}\Big|_{\lambda = 1} + \gamma = 0,
\end{equation*}
where $\Gamma$ and $\gamma$ are defined in equations \eqref{eq:Gamma_definition} and \eqref{eq:gamma_definition}, respectively. Assume $\Gamma$ is invertible. Then equation \eqref{eq:change_theta} follows.

Consistently estimating $\Gamma$ is straightforward because
\begin{equation*}
\pd{m_{it}}{\theta} =
\begin{pmatrix}
	\frac{\partial m_{it}}{\partial \theta_f} \\[7pt]
\frac{\partial m_{it}}{\partial \theta_g}
\end{pmatrix}
=
\begin{pmatrix}
	- \frac{\partial f(k_{it}, v_{it}; \theta_f)}{\partial \theta_f} + \frac{\partial g(\expt[q_{it-1}|x_{it-1}] - f(k_{it-1}, v_{it-1}; \theta_f); \theta_g)}{\partial \omega_{it-1} } \frac{\partial f(k_{it-1}, v_{it-1}; \theta_f)}{\partial \theta_f} \\[7pt]
 - \frac{\partial g(\expt[q_{it-1}|x_{it-1}] - f(k_{it-1}, v_{it-1}; \theta_f); \theta_g)}{\partial \theta_g}
\end{pmatrix}
\end{equation*}
and
\begin{equation*}
\pd{^2m_{it}}{\theta\partial \theta^{\top}} =
\begin{pmatrix}
	\pd{^2m_{it}}{\theta_f\partial \theta_f^{\top}}  & \pd{^2m_{it}}{\theta_f\partial \theta_g^{\top}} \\
\left(  \pd{^2m_{it}}{\theta_f\partial \theta_g^{\top}}  \right)^{\top} 	 & \pd{^2m_{it}}{\theta_g\partial \theta_g^{\top}}
\end{pmatrix},
\end{equation*}
where
\begin{eqnarray*}
	\pd{^2m_{it}}{\theta_f\partial \theta_f^{\top}} &=&
	- \frac{\partial^2 f(k_{it}, v_{it}; \theta_f)}{\partial \theta_f\partial \theta^\top_f} + \frac{\partial g(\expt[q_{it-1}|x_{it-1}] - f(k_{it-1}, v_{it-1}; \theta_f); \theta_g)}{\partial \omega_{it-1} } \frac{\partial^2 f(k_{it-1}, v_{it-1}; \theta_f)}{\partial \theta_f\partial \theta_f^\top} \\
							&& - \frac{\partial^2 g(\expt[q_{it-1}|x_{it-1}] - f(k_{it-1}, v_{it-1}; \theta_f); \theta_g)}{\partial \omega^2_{it-1} } \left( \frac{\partial f(k_{it-1}, v_{it-1}; \theta_f)}{\partial \theta_f}  \right)  \left( \frac{\partial f(k_{it-1}, v_{it-1}; \theta_f)}{\partial \theta_f}  \right)^\top,\\[7pt]
	\pd{^2m_{it}}{\theta_f\partial \theta_g^{\top}} &=&  \frac{\partial f(k_{it-1}, v_{it-1}; \theta_f)}{\partial \theta_f}  \frac{\partial^2 g(\expt[q_{it-1}|x_{it-1}] - f(k_{it-1}, v_{it-1}; \theta_f); \theta_g)}{\partial \omega_{it-1}\partial \theta_g^\top}, \\[7pt]
	\pd{^2m_{it}}{\theta_g\partial \theta_g^{\top}} &=&   \frac{\partial^2 g(\expt[q_{it-1}|x_{it-1}] - f(k_{it-1}, v_{it-1}; \theta_f); \theta_g)}{\partial \theta_g \partial \theta_g^\top},
\end{eqnarray*}
do not depend on the lagged prediction error $\zeta_{it-1}$. Following standard asymptotic arguments, $\Gamma$ can be consistently estimated by the finite-sample analog to equation \eqref{eq:Gamma_definition} using the estimate of $\expt[q_{it-1}|x_{it-1}]$ from step 1 and the estimate of $\theta$ from step 2.

\newpage
\bibliography{economic}
%\bibliography{myref}

\end{document}